\documentclass[journal,twoside,web]{ieeecolor}
\usepackage{generic}
\usepackage{cite}
\usepackage{algorithmic}
\usepackage{graphicx}
\usepackage[font=small,labelfont=bf]{caption}
\usepackage[labelformat=simple]{subcaption}

\usepackage{algorithm}
\usepackage{hyperref}
\hypersetup{hidelinks=true}
\usepackage{textcomp}
\def\BibTeX{{\rm B\kern-.05em{\sc i\kern-.025em b}\kern-.08em
    T\kern-.1667em\lower.7ex\hbox{E}\kern-.125emX}}
\definecolor{revcolor}{rgb}{0,0,0} 
\definecolor{arxivcolor}{rgb}{0,0,0} 

\usepackage{siunitx}

\usepackage{makecell}
\usepackage{booktabs}
\usepackage{multirow}

\usepackage{threeparttable}
\usepackage{array}
\newcolumntype{C}[1]{>{\centering\arraybackslash}m{#1}}
\newcolumntype{Y}{>{\centering\arraybackslash}X}

\usepackage{amsmath,amssymb,amsfonts}
\usepackage{dsfont}
\usepackage{mathtools}
\newtheorem{theorem}{Theorem}

\newtheorem{property}{Property}

\newtheorem{lemma}{Lemma}
\newtheorem{definition}{Definition}
\newtheorem{remark}{Remark}

\DeclareMathOperator{\Lip}{Lip}
\DeclareMathOperator{\diag}{diag}

\DeclareMathOperator{\vol}{vol}
\DeclareMathOperator{\cond}{cond}

\DeclareMathOperator{\sign}{sgn}
\newcommand{\bb}[1]{\mathbb{ #1 }}
\newcommand{\interior}{\mathrm{int}}
\newcommand{\grad}{\nabla}

\newcommand{\R}{\bb{R}}
\newcommand*{\defeq}{\stackrel{\text{\tiny def}}{=}}

\begin{document}
\title{Learning Robust Control Lyapunov Functions through Lipschitz Neural Networks}

\author{Shiqing~Wei,~\IEEEmembership{Graduate~Student~Member,~IEEE,} Prashanth~Krishnamurthy,~\IEEEmembership{Member,~IEEE,} \\ and Farshad~Khorrami,~\IEEEmembership{Fellow,~IEEE}
\thanks{This work was supported in part by ARO grants W911NF-21-1-0155 and W911NF-22-1-0028 and by the New York University Abu Dhabi (NYUAD) Center for Artificial Intelligence and Robotics (CAIR), funded by Tamkeen under the NYUAD Research Institute Award CG010.} 
\thanks{Shiqing Wei, Prashanth Krishnamurthy, and Farshad Khorrami are with Control/Robotics Research Laboratory, Electrical and Computer Engineering Department, Tandon School of Engineering, New York University, Brooklyn, NY 11201, USA. E-mail: \{shiqing.wei, prashanth.krishnamurthy, khorrami\}@nyu.edu.}}

\maketitle

\begin{abstract}
This work presents a novel framework for learning robust control Lyapunov functions and stabilizing controllers for nonlinear dynamical systems subject to additive disturbances upper bounded by a state-dependent function. We leverage recent advances in Lipschitz neural networks to jointly learn both the Lyapunov functions and state-feedback controllers. We establish explicit bounds on the Hessian and third-order derivatives of these neural networks in the spectral norm, and introduce a GPU-friendly branch-and-bound algorithm that utilizes higher-order bounds to significantly accelerate the verification of the Lyapunov conditions. Finally, we validate the proposed approach through extensive simulations on six different dynamical systems.
\end{abstract}

\begin{IEEEkeywords}
Stability of nonlinear systems, neural networks, robust control, computer-aided control design
\end{IEEEkeywords}

\section{Introduction}
\label{sec:introduction}
\IEEEPARstart{S}{tability} analysis of dynamical systems is crucial for controlling nonlinear systems where linear approximations are insufficient. In many scenarios, model mismatch or external disturbances may exist and necessitate robust control strategies, e.g., a robotic manipulator operating in an environment with unpredictable loads and friction variations.

Lyapunov functions are a foundational tool in the stability analysis of dynamical systems \cite{lyapunov1992general}. They offer a sufficient condition for the stability of a nonlinear system by designing a positive definite scalar function that decreases along solution trajectories, which enables characterization of the forward invariant region (or region of attraction) and robustification of control designs. While converse theorems guarantee the existence of a suitable Lyapunov function given certain stability properties of the system \cite{kellett2015classical}, these results are typically non-constructive in nature. Consequently, extensive research has focused on constructing Lyapunov functions using both analytical \cite{sepulchre2012constructive} and computational methods \cite{giesl2015review}.

Numerous computational methods have been investigated since the 1950s. Zubov introduced a specific Lyapunov function as the solution to a first-order partial differential equation (now known as Zubov's equation) whose solutions are computed using series expansions and other techniques \cite{zubov1961methods, vannelli1985maximal}. Later, Linear Programming (LP) methods were used to compute Lyapunov functions, first for switched linear systems \cite{brayton1979stability} and subsequently for nonlinear systems \cite{johansen2000computation}. Since the 1990s, significant effort has been dedicated to Linear Matrix Inequality (LMI) methods \cite{boyd1994linear} for nonlinear systems, which paved the way for a variety of techniques based on sums-of-squares (SOS) approaches \cite{parrilo2000structured, peet2009exponentially, peet2012converse}. However, the automatic construction of Lyapunov functions for general nonlinear systems remains an open question \cite{giesl2015review}.

Advances in computational hardware and deep learning frameworks have led many researchers to construct (control) Lyapunov functions for nonlinear systems using neural networks—a line of inquiry that also allows the joint learning of a state-feedback controller for a control system. Unlike earlier computational methods, this approach requires a formal verification of the Lyapunov conditions. Since Lyapunov functions are not necessarily convex \cite{jongeneel2024continuation}, this verification generally involves solving nonconvex optimization problems.
For continuous-time systems, recent work has utilized Satisfiability Modulo Theories (SMT) solvers such as dReal \cite{chang2019neural, zhou2022neural, liu2025physics} and Z3 \cite{abate2020formal}. Some other works use specific parameterization of the dynamics to bypass the verification of the Lyapunov function \cite{cheng2024learning, wang2024monotone}. In the discrete-time setting, optimization-based techniques have been employed \cite{dai2021lyapunov, wu2023neural}. In these studies, the neural networks use ReLU or leaky ReLU activation functions, which enable mixed-integer linear programming (MILP). In \cite{yang2024lyapunov}, bound propagation techniques such as $\alpha,\beta$-CROWN \cite{wang2021beta} are adopted to verify the learned Lyapunov function and controller without considering external disturbances.

Although recent research has shown that neural networks can effectively compute Lyapunov functions and feedback controllers, their robustness to external disturbances remains largely unexplored. For example, the studies in \cite{chang2019neural, liu2025physics, dai2021lyapunov, wu2023neural} assume exact dynamics without disturbances. In the following, we highlight some closely related works. In \cite{dawson2022safe}, the authors develop a learning framework to learn Lyapunov-barrier functions for control-affine systems with affine dependence in the uncertain parameters, although without formal verification. In \cite{zhou2022neural}, the authors tackle a scenario with partially known dynamics by first learning the system dynamics before proceeding to learn the Lyapunov function and controller. The work \cite{rego2022learning} considers input-to-state stability (ISS) Lyapunov functions and qualitatively examines the robustness of the learned controller to Gaussian noise. In contrast, we learn a robust Lyapunov function and a state-feedback controller for nonlinear control systems with additive disturbances upper bounded by a function of the state in continuous time. Accounting for such disturbances introduces additional challenges. First, the Lyapunov condition must be reformulated, which complicates the verification process. Second, the Lyapunov function must decrease uniformly across all admissible disturbances, making both training and verification significantly more demanding. We also notice that the choice of neural networks in existing works relies heavily on the capability of the verification method. For instance, the works \cite{chang2019neural, liu2025physics, zhou2022neural} use neural networks of two layers. Since Lyapunov functions can be scaled by a positive constant without affecting the stability analysis, we employ Lipschitz neural networks (LNNs) \cite{miyato2018spectral, trockman2021orthogonalizing, wang2023direct} to learn both Lyapunov functions and state-feedback controllers. In our approach, we derive bounds on the Hessian and third-order derivatives of LNNs with respect to their input. The key motivation is that Lyapunov conditions involve nonlinear expressions, and methods based only on zeroth- or first-order information typically lead to overly conservative bounds. By incorporating Hessian and third-order derivative bounds through higher-order Taylor expansions, we obtain tighter bounds on the network output, which directly enhances the precision and effectiveness of the verification process. Building on this insight, we introduce a branch-and-bound (BnB) algorithm to verify the Lyapunov functions. Compared to dReal \cite{gao2013dreal}, which uses only zeroth-order information and runs solely on CPUs, our BnB algorithm supports higher-order bounds and batch evaluation on GPUs.

\textit{Our contributions:} (1) We propose a novel framework to learn robust control Lyapunov functions and stabilizing state-feedback controllers for nonlinear control systems through LNNs. (2) We establish bounds on the input Hessian and third-order derivative of LNNs in the spectral norm, providing tighter bounds than zeroth- or first-order information. (3) Building on these higher-order bounds, we develop a GPU-friendly branch-and-bound-based verification algorithm that significantly reduces the verification time of the learned Lyapunov functions. We establish the worst-case time complexity and validate our approach through simulations on six different dynamical systems, including the inverted pendulum, the unicycle path following example, a third-order system, the cartpole, the 2D quadrotor, and the SCARA arm.

\section{Notations} \label{sec:notations}
Let $\mathbb{R}$ denote real numbers and $\mathbb{D}_{++}^n$ the set of $n \times n$ positive definite diagonal matrices. $\lVert \cdot \rVert_p$ is the $\ell_p$ norm of a vector or the induced norm of a matrix. $I$ is the identity matrix. $\diag(\cdot)$ takes a vector as input and returns a diagonal matrix with that vector on the diagonal. A class $\mathcal{C}^k$ function is a function that has a continuous $k$-th order derivative. A class $\mathcal{K}$ function $\alpha: [0, a) \rightarrow [0,\infty)$ is a continuous strictly increasing function with $\alpha(0) = 0$. For a $\mathcal{C}^1$ function $h: \R^n \rightarrow \R$, its gradient $\grad h$ is a column vector while the partial derivative $\frac{\partial h}{\partial x}$ is a row vector. For a $\mathcal{C}^1$ function $f: \R^n \rightarrow \R^m$, $\frac{\partial f}{\partial x}$ is an $m \times n$ matrix. For a set $A$, $\partial A$ denotes the boundary of $A$, and $\interior (A)$ the interior of $A$. Denote by $B_p(a, r) \subset \R^n$ the $\ell_p$ ball centered at $a \in \R^n$ with a radius of $r > 0$. $\mathds{1}_A (\cdot)$ is the indicator function, and $\mathds{1}_A (x) = 1$ if $x \in A$ and $\mathds{1}_A (x) = 0$ otherwise. For two matrices $A$ and $B$ of the same dimension, $\odot$ represents the Hadamard product and $(A \odot B)_{ij} = A_{ij} B_{ij}$. For $x \in \R$, $\lceil x \rceil$ returns the least integer greater than or equal to $x$.

\section{Preliminaries}\label{sec:preliminaries}
{\color{revcolor} Many neural networks can be seen as a sequential composition of different layers. We briefly introduce a generalized class of neural network layers with a prescribed Lipschitz constant. First, recall the definition of Lipschitz continuity.

\begin{definition}\label{def:lip_continuous}
    A function $\phi: \R^n \rightarrow \R^m$ is Lipschitz continuous (or $\gamma$-Lipschitz), if there exists a positive real constant $\gamma$ such that for all $x_1, x_2 \in \R^n$
    \begin{equation}\label{eq:def_lip}
        \lVert \phi(x_1)-\phi(x_2) \rVert_2 \leq \gamma \lVert x_1 - x_2 \rVert_2.
    \end{equation}
    The smallest $\gamma$ such that \eqref{eq:def_lip} holds is called the Lipschitz constant of $\phi$, which is denoted by $\Lip(\phi)$.
\end{definition}

Denote by $h_{\mathrm{in}} \in \R^{n_{\mathrm{in}}}$ and $h_{\mathrm{out}} \in \R^{n_{\mathrm{out}}}$ the input and output of a certain layer, respectively. A generalized neural network layer can be formulated as
\begin{equation}\label{eq:generalized_layer}
    h_{\mathrm{out}} = U \sigma (W h_{\mathrm{in}} + b)
\end{equation}
where $U \in \R^{n_{\mathrm{out}} \times n_{\mathrm{out}}}$ and $W \in \R^{n_{\mathrm{out}} \times n_{\mathrm{in}}}$ are weight matrices, $b$ is the bias vector, and $\sigma: \R \rightarrow \R$ is a $\mathcal{C}^0$ activation function (applied elementwise) with slope restricted to $[0,1]$. Examples of such activation functions include the softplus $\sigma(x) = \log(1+e^x)$, the sigmoid $\sigma(x) = 1/(1+e^{-x})$, the hyperbolic tangent $\sigma(x) = \tanh(x)$, the rectified linear unit (ReLU) $\sigma(x) = \max(0,x)$, etc. 

Several representative parameterizations have been proposed in the literature to rigorously enforce this 1-Lipschitz property during training:
\begin{itemize}
    \item \textbf{Spectral Normalization \cite{miyato2018spectral}:} Sets $U = I$ and $W = \widetilde{W} / \lVert \widetilde{W} \rVert_2$, where $\widetilde{W} \in \R^{n_{\mathrm{out}} \times n_{\mathrm{in}}}$ is a standard weight matrix and $\lVert \widetilde{W} \rVert_2$ is its spectral norm;
    \item \textbf{Orthogonal Layers \cite{trockman2021orthogonalizing}:} Sets $U = I$ and parameterizes $W$ to be semi-orthogonal (i.e., $W^\top W = I$ or $W W^\top = I$, yielding $\lVert W \rVert_2 = 1$) through the Cayley transform; 
    \item \textbf{Sandwich Layers \cite{wang2023direct}:} Sets $U = \sqrt{2}A^\top \Psi$ and $W = \sqrt{2}\Psi^{-1}B$, where $\Psi = \diag([e^{v_1}, \dots, e^{v_{n_{\mathrm{out}}}}]) \in \mathbb{D}_{++}^{n_{\mathrm{out}}}$ is a learnable diagonal scaling matrix, $A \in \R^{n_{\mathrm{out}} \times n_{\mathrm{out}}}$ and $B \in \R^{n_{\mathrm{out}} \times n_{\mathrm{in}}}$ are also parameterized through the Cayley transform in \cite{trockman2021orthogonalizing} such that $AA^\top + BB^\top = I$.
\end{itemize}

Building upon this generalized layer, we consider the following LNN with $L \geq 2$ layers:
{\allowdisplaybreaks
\begin{subequations}\label{eq:lip_nn}
\begin{align}
    h^{(0)} &= \sqrt{\gamma} x, \\
    h^{(\ell)} &= U^{(\ell)} \sigma^{(\ell)} \Big( \underbrace{W^{(\ell)} h^{(\ell-1)} + b^{(\ell)}}_{z^{(\ell)}} \Big), \quad \ell = 1, ..., L-1, \label{eq:first_layers}\\
    F(x) &\defeq h^{(L)} = \sqrt{\gamma} W^{(L)} h^{(L-1)}  \label{eq:last_layer}
\end{align}
\end{subequations}
}where $\gamma > 0$ defines the global Lipschitz bound of the network. The $L-1$ hidden layers defined in \eqref{eq:first_layers} are 1-Lipschitz layers parameterized via any of the methods discussed above, and the final linear layer defined in \eqref{eq:last_layer} satisfies $\lVert W^{(L)} \rVert_2 \leq 1$. Denote by $n_{\ell}$ the dimension of $h^{(\ell)}$ (for $\ell = 0, ..., L$) with $n_{0} = n_x$. Note that $x$ and $h^{(L)}$ are the input and output of this neural network, respectively, and we have omitted the dependence of $h^{(\ell)}$ (for $\ell = 0, ..., L$) on $x$ in \eqref{eq:lip_nn} for brevity. The shorthand notation $z^{(\ell)}$ is used for later analysis. To enforce the equilibrium condition required for Lyapunov analysis, we define the final network output as:
\begin{equation}\label{eq:shifted_lip_nn}
    \hat{F}(x) \defeq F(x) - F(0).
\end{equation}
Because $F(x)$ is constructed to be $\gamma$-Lipschitz \cite[Section~3.2]{szegedy2013intriguing}, it trivially follows that the shifted network $\hat{F}(x)$ is also $\gamma$-Lipschitz and satisfies $\hat{F}(0) = 0$.
}

\section{Learning of Robust Control Lyapunov Functions}\label{sec:learning}
\subsection{Robust Control Lyapunov Functions}
Consider the following dynamical system
\begin{equation} \label{eq:sys}
    \dot{x} = f(x,u) + G d(t,x)
\end{equation}
where $x \in \mathcal{X} \subset \R^{n_x}$ is the state, $u \in \mathcal{U} \subset \R^{n_u}$ the control input, and $\mathcal{X}$ and $\mathcal{U}$ are compact subsets of $\R^{n_x}$ and $\R^{n_u}$, respectively. We assume that $f: \R^{n_x} \times \R^{n_u} \rightarrow \R^{n_x}$ is a $\mathcal{C}^2$ function in $x$ and $u$, and there exists $u_0 \in \mathcal{U}$ such that $f(0, u_0) = 0$. The function $d: \R \times \R^{n_x} \rightarrow \R^{n_d}$ represents a disturbance signal that can be time- and state-dependent, and we assume that $\lVert d(t,x) \rVert_2 \leq \epsilon(x)$ where $\epsilon: \R^{n_x} \rightarrow \R_+$ is known and Lipschitz continuous in $x$. The matrix $G \in \R^{n_x \times n_d}$ specifies the channel of the disturbance signal $d(t,x)$. We are interested in finding a state-feedback controller $\pi: \R^{n_x} \rightarrow \R^{n_u}$ that stabilizes \eqref{eq:sys}. The closed-loop system is
\begin{equation} \label{eq:cl_sys}
    \dot{x} = f(x,\pi(x)) + G d(t,x) \defeq f_{\mathrm{cl}}(t,x).
\end{equation}

Recall the definition of uniform ultimate boundedness:
\begin{definition}[Definition~4.6 \cite{khalil2002nonlinear}]\label{def:uub}
The solution to \eqref{eq:cl_sys} is said to be uniformly ultimately bounded (UUB) if there exist two positive constants $b$ and $c$ (both independent of $t_0$) such that, for any $a \in (0,c)$, there exists $T=T(a,b) \geq 0$ (independent of $t_0$) such that
\begin{equation}\label{eq:uub}
    \lVert x(t_0) \rVert_2 \leq a \Rightarrow \lVert x(t) \rVert_2 \leq b, \quad \forall t \geq t_0 +T.
\end{equation}
In addition, $b$ is referred to as the ultimate bound.
\end{definition}

Next, we introduce the definition of robust control Lyapunov functions used in this work.
\begin{definition}\label{def:rclf}
Let $D \subset \R^{n_x}$ be a domain containing the origin and $V: D \to \R$ a $\mathcal{C}^1$ function. $V$ is a robust control Lyapunov function (RCLF) for the system \eqref{eq:cl_sys} if there exists $\mu > 0$ such that:
\begin{enumerate}
\item there exists $V_r>0$ such that $B_2(0,\mu) \subset \Omega_{V_r} \subset D$ where $\Omega_{V_r} = \{ x \in \R^{n_x} \mid V(x) \leq V_r \}$.

\item there exist class $\mathcal{K}$ functions $\alpha_1$ and $\alpha_2$, and a Lipschitz continuous positive definite function $\omega$ such that for all $t \geq t_0$ and $x \in D$
\begin{subequations}
\begin{align}
    &\alpha_1(\lVert x \rVert_2) \leq V(x) \leq \alpha_2(\lVert x \rVert_2), \label{eq:pd_cond}\\
    &\dot{V}(t,x) \defeq \frac{\partial V}{\partial x}(x) f_{\mathrm{cl}}(t, x) \leq -\omega(x),\; \forall \lVert x \rVert_2 \geq \mu .\label{eq:stability_cond}
\end{align}
\end{subequations}
\end{enumerate}
\end{definition}

The following properties hold for an RCLF.
\begin{property}\label{prpt:forward_invariant}
If $V$ is an RCLF for system \eqref{eq:cl_sys}, then $\Omega_{V_r} = \{ x \in \R^{n_x} \mid V(x) \leq V_r \}$ is forward invariant. 
\end{property}
\begin{proof}
As $\dot{V}(t,x) < 0$ on $D \backslash B_2(0,\mu)$ by \eqref{eq:stability_cond} and $B_2(0,\mu) \subset \Omega_{V_r}$ (by the first condition of Definition~\ref{def:rclf}), we know that $\dot{V}(t,x) < 0$ on $\partial \Omega_{V_r}$. Therefore, $\Omega_{V_r}$ is forward invariant by Nagumo's Theorem \cite[Theorem~4.7]{blanchini2015set}.
\end{proof}

\begin{property}\label{prpt:uub}
If $V$ is an RCLF for system \eqref{eq:cl_sys}, then the solution to \eqref{eq:cl_sys} is UUB with $c = \alpha_2^{-1}(V_r)$ and 
\begin{equation}
b = \begin{cases}
    \alpha_1^{-1}( \alpha_2 (\mu) ), & \text{if} \ c > \alpha_1^{-1}( \alpha_2 (\mu) ),\\
    \alpha_1^{-1}(V_r), & \text{otherwise}.
\end{cases}
\end{equation}

\end{property}
\begin{proof}
Let $x(t_0)$ be the initial state of $\eqref{eq:cl_sys}$, and define $c = \alpha_2^{-1}(V_r)$, $b_1 = \alpha_1^{-1}( \alpha_2 (\mu) )$, and $b_2 = \alpha_1^{-1}(V_r)$. From \eqref{eq:pd_cond}, we have $\alpha_1 (\mu) \leq \alpha_2 (\mu)$, which implies $\mu \leq \alpha_1^{-1}( \alpha_2 (\mu) ) = b_1$. Moreover, \eqref{eq:pd_cond} ensures that $V(x) \leq \alpha_2 (c) = V_r$ for all $\lVert x \rVert_2 \leq c$, so $B_2(0,c) \subset \Omega_{V_r}$. If $c > b_1$, then $\mu \leq b_1 < c$, and thus $B_2(0, \mu) \subset B_2(0,b_1) \subset B_2(0,c) \subset \Omega_{V_r}$. Next, we distinguish two cases. If $V(x(t_0)) \leq \alpha_2 (\mu)$, then \eqref{eq:uub} holds with $T=0$ and $b=b_1$, because $\{ x \in \R^{n_x} \mid V(x) \leq \alpha_2 (\mu) \}$, which contains $B_2(0, \mu)$, is forward invariant by \eqref{eq:stability_cond}. Consequently, $\lVert x(t) \rVert_2 \leq b_1$ for all $t \geq t_0$. Otherwise, $V(x(t_0)) > \alpha_2 (\mu)$. Take $k=\min_{x\in A}\omega(x)$ with $A = \{ x \in \Omega_{V_r} \mid \lVert x \rVert_2 \geq \mu \}$. Then, $V(x(t)) \leq V(x(t_0)) - k (t-t_0)$ holds on $A$ by \eqref{eq:stability_cond}, which shows that $V(x(t))$ must decrease to $\alpha_2 (\mu)$ within the time interval $[t_0, t_0 + T]$ with $T = (V(x(t_0)) - \alpha_2 (\mu))/k$. On the other hand, if $c \leq b_1$, then from the forward invariance of $\Omega_{V_r}$ by Property~\ref{prpt:forward_invariant}, \eqref{eq:uub} holds with $T=0$ and $b=b_2$. This concludes the proof.
\end{proof}

\begin{property}\label{prpt:scaling}
If $V$ is an RCLF for system \eqref{eq:cl_sys}, then $\kappa V$ is also an RCLF for system \eqref{eq:cl_sys} for any $\kappa > 0$.
\end{property}
\begin{proof}
It suffices to replace $V_r$ with $\kappa V_r$ in the first condition of Definition~\ref{def:rclf}, and $\alpha_1$, $\alpha_2$, and $\omega$ with $\kappa \alpha_1$, $\kappa \alpha_2$, and $\kappa \omega$, respectively, in the second condition.
\end{proof}

Property~\ref{prpt:scaling} implies that, on a bounded domain, it is possible to limit the Lipschitz constant of an RCLF. In the following, we will use neural networks to learn an RCLF and a stabilizing state-feedback controller.

\begin{remark}
In the absence of disturbances, setting $\mu=0$ in Definition~\ref{def:rclf} reduces our formulation to the standard CLF. In this case, the UUB result in Property~\ref{prpt:uub} strengthens to asymptotic stability, thereby recovering the classical CLF framework (e.g., \cite{artstein1983stabilization}, \cite[Theorem 4.8]{khalil2002nonlinear}) and connecting directly to prior neural CLF approaches (e.g., \cite{chang2019neural}).
\end{remark}

\subsection{Learning Algorithm}
\subsubsection{Setup}
We use LNNs to train an RCLF and a stabilizing controller. Let $V_{\theta_1}$ be a shifted LNN defined in \eqref{eq:shifted_lip_nn} for the RCLF, and denote by $L_V$ and $\gamma_V$ the number of layers and Lipschitz bound of $V_{\theta_1}$, respectively. The activation functions of $V_{\theta_1}$ are chosen such that $V_{\theta_1}$ is $\mathcal{C}^3$. Let $\pi_{\theta_2}$ also be an LNN for the controller and denote by $L_\pi$ and $\gamma_\pi$ the number of layers and Lipschitz bound of $\pi_{\theta_2}$, respectively. We only require the activation functions of $\pi_{\theta_2}$ to be $\mathcal{C}^0$. In addition, we consider box constraints on the control inputs for the rest of this work, i.e., $\mathcal{U} = \{ u \in \R^{n_u} \mid \underline{u}_i \leq u_i \leq \overline{u}_i, i= 1,..., n_u \}$, and the final outputs of $\pi_{\theta_2}$ are clipped between $\underline{u}_i$ and $\overline{u}_i$ (which does not increase its Lipschitz bound). $\theta_1$ and $\theta_2$ are the trainable parameters of $V_{\theta_1}$ and $\pi_{\theta_2}$, respectively. 

LNNs are adopted in this work because a predefined Lipschitz constant enables explicit bounds on higher-order derivatives, which yield tighter bounds and significantly accelerate verification of the Lyapunov function and controller. In addition, the Lipschitz constraint serves as a natural regularizer that improves generalization during training. Although this choice reduces expressivity compared with unconstrained networks, the trade-off is justified since the derivative bounds are essential for efficient verification. Moreover, an RCLF can always be scaled so its Lipschitz constant falls within a prescribed bound (Property~\ref{prpt:scaling}), and for state-feedback controllers, larger Lipschitz bounds can be used to mitigate potential conservatism.

\subsubsection{Training Process}
We consider a bounded state space $\mathcal{X}$. Define the training set $\mathcal{D}_1 = \{x^{(1)}, ..., x^{(N_1)}\}$ that is sampled from $\mathcal{X}$ and the boundary set $\mathcal{D}_2 = \{x^{(1)}, ..., x^{(N_2)}\}$ that is sampled from $\partial \mathcal{X}$. In the following, we define different loss terms in order to train $V_{\theta_1}$ and $\pi_{\theta_2}$ to be a valid RCLF and a stabilizing controller, respectively.

\textbf{Positive definiteness:} To encourage $V_{\theta_1}$ to be positive definite (i.e., condition \eqref{eq:pd_cond}), we define the following loss:
\begin{equation}
    \mathcal{L}_1 = \frac{1}{N_1} \sum_{x^{(i)} \in \mathcal{D}_1} \max \left(- V_{\theta_1}(x^{(i)}) + \varepsilon_1, 0 \right).
\end{equation}
where $\varepsilon_1 > 0$ is intended to create some margin as we only train on a finite number of data samples.

\textbf{Decrease condition:} For condition \eqref{eq:stability_cond}, we first upper bound $\dot{V}(t,x)$ by
\begin{align} \label{eq:def_H}
    &\dot{V}_{\theta_1}(t,x) = \frac{\partial V_{\theta_1}}{\partial x} \left( f(x,\pi_{\theta_2}(x)) + G d(t,x) \right) \notag \\
    &\leq \frac{\partial V_{\theta_1}}{\partial x} f(x,\pi_{\theta_2}(x)) \allowbreak + \left \lVert \frac{\partial V_{\theta_1}}{\partial x} G \right \rVert_2 \epsilon(x) \defeq H(x).
\end{align}
Then, to encourage $H(x) + \omega(x) \leq 0$ for $\lVert x \rVert_2 \geq \mu$ (which is a stronger condition than the original condition \eqref{eq:stability_cond}), we define the loss:
\begin{align}
    &\mathcal{L}_2 = \frac{1}{\sum_{x^{(i)} \in \mathcal{D}_1} \mathds{1}_{\{\lVert x \rVert_2 \geq \mu \}}(x^{(i)})}  \times \notag\\
    &\sum_{x^{(i)} \in \mathcal{D}_1} \max \left( H(x^{(i)}) + \omega (x^{(i)}) + \varepsilon_2, 0 \right) \mathds{1}_{\{\lVert x \rVert_2 \geq \mu \}}( x^{(i)})
\end{align}
where $\varepsilon_2 > 0$ is used to create some margin.

We use $\widehat{V}_r = \min_{x \in \mathcal{D}_2} V_{\theta_1} (x)$ (an estimate of $V_r$) to identify the largest possible forward invariant region entirely contained in $\mathcal{X}$ (which will be verified later in Section~\ref{sec:verification}). The next two losses are designed to increase the size of this region.

\textbf{Boundary value:} To encourage that a ball $B_2(0, \mu)$ can be contained within a sublevel set of $V_{\theta_1}$, we define
\begin{equation}
    \mathcal{L}_3 = \frac{1}{N_2} \sum_{x^{(i)} \in \mathcal{D}_2} \max \left( -V_{\theta_1} (x^{(i)}) + V_l + \varepsilon_3, 0 \right)
\end{equation}
where $V_l > 0$ serves as an intended lower bound of $V_{\theta_1}$ on $\partial \mathcal{X}$ and $\varepsilon_3 > 0$ creates some margin.

\textbf{Maximal coverage:} To encourage a larger coverage of $\mathcal{X}$ by a sublevel set of $V_{\theta_1}$, we define the following loss
\begin{equation}
    \mathcal{L}_4 = \sqrt{ \frac{1}{N_2 - 1} \sum_{x^{(i)} \in \mathcal{D}_2} \left( V_{\theta_1}(x^{(i)}) - \overline{V} \right)^2 }
\end{equation}
where $\overline{V} = \frac{1}{N_2} \sum_{x^{(i)} \in \mathcal{D}_2} V_{\theta_1}(x^{(i)})$ to reduce the standard deviation of $V_{\theta_1}$ on $\partial \mathcal{X}$. Intuitively, if $V_{\theta_1}$ has a small variation on $\partial \mathcal{X}$, then $\{ x \in \mathcal{X} \mid  V_{\theta_1}(x) \leq \widehat{V}_r \}$ is more likely to cover a larger portion of $\mathcal{X}$ and $\widehat{V}_r$ will be less conservative.

\begin{algorithm}[t]
\caption{Training of $V_{\theta_1}$ and $\pi_{\theta_2}$}
\label{alg:train}
\begin{algorithmic}[1]
{\color{revcolor}
\REQUIRE Dynamical system $f$, bound on the disturbance signal $\epsilon$, state space $\mathcal{X}$, neural networks $V_{\theta_1}$ and $\pi_{\theta_2}$, hyperparameters $V_l, \mu, c_1, c_2, c_3, c_4, \varepsilon_1, \varepsilon_2, \varepsilon_3$, and number of epochs $N_e$.

\STATE Pretrain $V_{\theta_1}$ and $\pi_{\theta_2}$ to LQR solutions;
\STATE Sample $\mathcal{X}$ and $\partial \mathcal{X}$ to obtain $\mathcal{D}_1$ and $\mathcal{D}_2$;

\FOR{epoch $i = 1, \dots, N_e$}
    \STATE Update $\theta_1$ and $\theta_2$ via gradient descent to minimize the loss $\mathcal{L}$ over mini-batches of $\mathcal{D}_1$ and $\mathcal{D}_2$;
    \STATE Evaluate empirical losses $\mathcal{L}'_1$ and $\mathcal{L}'_2$ on $\mathcal{D}_1$ by setting margins $\varepsilon_1 = \varepsilon_2 = 0$;

    \IF{$\mathcal{L}'_1 = 0$ and $\mathcal{L}'_2 = 0$}
        \STATE Calculate $\widehat{V}_r = \min_{x \in \mathcal{D}_2} V_{\theta_1}(x)$;
        \STATE Store $\theta_1^\star$, $\theta_2^\star$, and $\widehat{V}_r^\star$ if the current parameters yield the largest empirical forward invariant region computed so far;
    \ENDIF
\ENDFOR

\RETURN $\theta_1^\star$, $\theta_2^\star$, and $\widehat{V}_r^\star$.
}
\end{algorithmic}
\end{algorithm}

The overall loss function is a weighted sum of these losses:
\begin{equation}\label{eq:total_loss}
\mathcal{L} = c_1 \mathcal{L}_1 + c_2 \mathcal{L}_2 + c_3 \mathcal{L}_3 + c_4 \mathcal{L}_4
\end{equation}
with $c_1, c_2, c_3, c_4 > 0$. The detailed training algorithm is given in Algorithm~\ref{alg:train}. It first initializes $V_{\theta_1}$ and $\pi_{\theta_2}$ using the LQR (Linear Quadratic Regulator) solution based on a linearization of the system, then trains $V_{\theta_1}$ and $\pi_{\theta_2}$ to satisfy conditions \eqref{eq:pd_cond} and \eqref{eq:stability_cond}, and keeps the optimal parameters that yield the largest forward invariant region. Algorithm~\ref{alg:train} is most likely to yield a valid RCLF when (i) the architecture of $V_{\theta_1}$ is expressive enough to approximate a Lyapunov function and the Lipschitz bound of $\pi_{\theta_2}$ is sufficiently large to provide control authority, (ii) the training dataset adequately covers the relevant state space, and (iii) conservative offsets $\varepsilon_1, \varepsilon_2$ are used initially. Since Algorithm~\ref{alg:train} minimizes the loss $\mathcal{L}$ in \eqref{eq:total_loss} only on a finite set of sampled states, eliminating all violation points requires parameter tuning. A common strategy is to begin with conservative values of $\varepsilon_1$ and $\varepsilon_2$, and then gradually reduce them to prevent overfitting of $V_{\theta_1}$ and $\pi_{\theta_2}$.

\begin{remark}\label{rmk:learned_dynamics}
{\color{revcolor} We emphasize that our framework does not certify stability for an unknown model: as is well known, formal certification over a domain requires reliable knowledge of the dynamics. What the disturbance-robust formulation provides is \emph{tolerance to a bounded, certified model mismatch}. Suppose the true dynamics $\dot{x}=f(x,u)$ are not known exactly, but a nominal model $\hat{f}(x,u)$ is available together with a known, rigorous error bound on $f-\hat{f}$ over the domain $D$ (e.g., from robust system identification, or Gaussian-process
regression with high-probability confidence bounds). Writing the closed loop as $\dot{x}=\hat{f}(x,\pi_{\theta_2}(x))+d(x)$ with mismatch
$d(x)=f(x,\pi_{\theta_2}(x))-\hat{f}(x,\pi_{\theta_2}(x))$, any certified bound $\lVert d(x)\rVert_2 \le \epsilon(x)$ allows $d$ to be treated as the state-dependent disturbance of \eqref{eq:sys} (with $G=I$), so that Algorithm~\ref{alg:train} can be applied to the nominal $\hat{f}$. For instance, if Lipschitz constants $K_x,K_u$ of
$f-\hat{f}$ in $x$ and $u$ are known on $D$ and $f(0,\pi_{\theta_2}(0)) - \hat{f}(0,\pi_{\theta_2}(0)) = 0$, then $\lVert d(x)\rVert_2 \le (K_x+K_u\gamma_\pi)\lVert x\rVert_2 \defeq \epsilon(x)$. The resulting guarantee is exactly as strong as the model-error bound itself: deterministic when $\epsilon$ is a hard bound, and high-probability when $\epsilon$ holds with a given confidence. Quantifying such a bound is a separate identification problem and is beyond the scope of this paper.}
\end{remark}

{\color{revcolor}
\begin{remark}\label{rmk:bounding_K_functions}
Condition \eqref{eq:pd_cond} requires class $\mathcal{K}$ functions $\alpha_1$ and $\alpha_2$ such that $\alpha_1(\lVert x \rVert_2) \leq V(x) \leq \alpha_2(\lVert x \rVert_2)$ for all $x \in D$. Because the neural network parameterization ensures that $V$ is continuous and the verification domain $D \subset \mathbb{R}^{n_x}$ is compact and contains the origin, explicitly enforcing positive definiteness ($V(0)=0$ and $V(x) > 0$ for $x \in D \setminus \{0\}$) guarantees, by extending \cite[Lemma~4.3]{khalil2002nonlinear} from a ball $B_2(0,r) \subset D$ to all of $D$ via compactness, that valid class $\mathcal{K}$ bounding functions exist \textit{a posteriori} for all $x \in D$.
\end{remark}
}

\begin{figure}[t]
    \centering
    \includegraphics[width=0.85\linewidth]{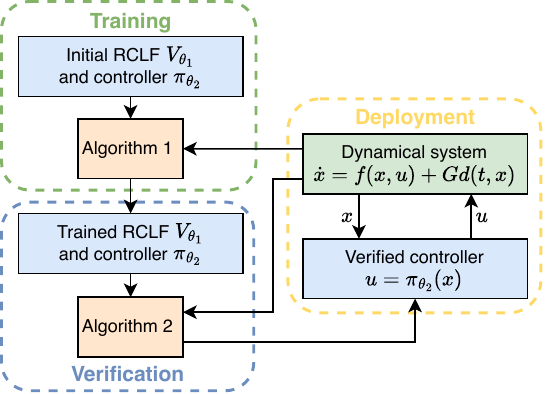}
    \caption{Overall control-design workflow.}
    \label{fig:diagram}
\end{figure}

\section{Verification of Robust Control Lyapunov Functions}\label{sec:verification}
In this section, we develop an algorithm to verify the trained state-feedback controller and RCLF. The complete workflow—from training to verification and deployment—is illustrated in Fig.~\ref{fig:diagram}. The conditions in Definition~\ref{def:rclf} are generally nonlinear and nonconvex, so tighter bounds are needed to reduce conservatism and improve efficiency. Leveraging the structure and explicit Lipschitz bounds of LNNs, we can compute higher-order derivative bounds in a tractable way.

\subsection{Bounds on the Hessian and Third-Order Derivative of LNNs}\label{sec:bounds}

{\color{revcolor} \begin{lemma}\label{lemma:lip_nn_jac}
Consider the LNN $F(x)$ defined in \eqref{eq:lip_nn}. If each activation function $\sigma^{(\ell)}$ is $\mathcal{C}^1$ for $\ell = 1, ..., L-1$, then the Jacobian of $F(x)$ w.r.t. $x$ is 
\begin{equation}\label{eq:lip_nn_jac}
\frac{\partial F(x)}{\partial x} = \frac{\partial F(x)}{\partial h^{(L-1)}} \frac{\partial h^{(L-1)}}{\partial h^{(L-2)}} \cdots \frac{\partial h^{(1)}}{\partial h^{(0)}}
\frac{\partial h^{(0)}}{\partial x}
\end{equation}
where $\frac{\partial h^{(0)}}{\partial x} = \sqrt{\gamma} I$, $\frac{\partial F(x)}{\partial h^{(L-1)}} = \sqrt{\gamma} W^{(L)}$, and $\frac{\partial h^{(\ell)}}{\partial h^{(\ell-1)}} = U^{(\ell)} \diag ({\sigma^{(\ell)}}'(z^{(\ell)})) W^{(\ell)}$ for $\ell = 1,..., L-1$. Furthermore, the shifted network satisfies $\frac{\partial \hat{F}(x)}{\partial x} = \frac{\partial F(x)}{\partial x}$.
\end{lemma}
\begin{proof}
    This is a direct application of the chain rule.
\end{proof}}

{\color{revcolor}
\begin{lemma}\label{lemma:lip_nn_hess}
Consider the LNN $F(x)$ defined in \eqref{eq:lip_nn}. If $\sigma^{(\ell)}$ is $\mathcal{C}^2$ for $\ell = 1, ..., L-1$, then the Hessian of $F_i(x)$ (the $i$-th component of $F(x)$) w.r.t. $x$ is
\begin{align}\label{eq:lip_nn_hess}
\frac{\partial^2 F_i(x)}{\partial x^2} &= \sum_{\ell=1}^{L-1} \Bigl ( W^{(\ell)} \frac{\partial h^{(\ell-1)}}{\partial x} \Bigr )^\top \diag \Bigl ( {\sigma^{(\ell)}}''(z^{(\ell)}) \notag \\
&\quad \quad  \quad  \quad  \odot  \Bigl (\frac{\partial F_i(x)}{\partial h^{(\ell)}} U^{(\ell)} \Bigr )^\top  \Bigr ) W^{(\ell)} \frac{\partial h^{(\ell-1)}}{\partial x}.
\end{align}
\end{lemma}

\begin{proof}
The proof is given in Appendix~\ref{appx:proof_lip_nn_hess}.
\end{proof}}

From Lemma~\ref{lemma:lip_nn_hess}, we are ready to establish a Hessian bound of the LNNs in the spectral norm.

{\color{revcolor}
\begin{theorem}\label{thm:lip_nn_hess_bound}
Consider the LNN $F(x)$ defined in \eqref{eq:lip_nn}. If $\sigma^{(\ell)}$ is $\mathcal{C}^2$ for $\ell = 1, \dots, L-1$, then the spectral norm of the Hessian $\frac{\partial^2 F_i(x)}{\partial x^2}$ satisfies
\begin{equation}\label{eq:lip_nn_hess_bound}
\left \lVert \frac{\partial^2 F_i(x)}{\partial x^2} \right \rVert_2 \le \sum_{\ell=1}^{L-1} \gamma^{\frac{3}{2}} \sup\bigl|{\sigma^{(\ell)}}''\bigr| \bigl\lVert \diag\bigl(c (U^{(\ell)})\bigr)^{1/2} W^{(\ell)} \bigr\rVert_2^2
\end{equation}
where $[c(U^{(\ell)})]_k = \lVert U^{(\ell)}_{:,k} \rVert_2$.
\end{theorem}

\begin{proof}
By Lemma~\ref{lemma:lip_nn_hess}, $\frac{\partial^2 F_i(x)}{\partial x^2} = \sum_{\ell=1}^{L-1} \Theta^{(\ell)\top} \allowbreak \diag(\lambda^{(\ell)}) \Theta^{(\ell)}$, where $\Theta^{(\ell)} = W^{(\ell)} \frac{\partial h^{(\ell-1)}}{\partial x}$ and $\lambda^{(\ell)} = {\sigma^{(\ell)}}''(z^{(\ell)}) \odot \bigl(\frac{\partial F_i(x)}{\partial h^{(\ell)}} U^{(\ell)}\bigr)^\top$. Given a fixed $\ell$ and any unit vector $v$,
$\lvert v^\top \Theta^{(\ell)\top} \diag(\lambda^{(\ell)}) \Theta^{(\ell)} v \rvert = \lvert \sum_k \lambda^{(\ell)}_k (\Theta^{(\ell)} v)_k^2 \rvert \le \sum_k \lvert \lambda^{(\ell)}_k \rvert (\Theta^{(\ell)} v)_k^2,$
hence $\lVert \Theta^{(\ell)\top} \diag(\lambda^{(\ell)}) \Theta^{(\ell)} \rVert_2 \le \lVert \Theta^{(\ell)\top} \diag(|\lambda^{(\ell)}|) \Theta^{(\ell)} \rVert_2$. By the Cauchy--Schwarz inequality and $\big \lVert \frac{\partial F_i(x)}{\partial h^{(\ell)}} \big \rVert_2 \le \sqrt{\gamma}$, we have 
$\lvert \lambda^{(\ell)}_k \rvert \le \sqrt{\gamma}\,\sup \lvert {\sigma^{(\ell)}}'' \rvert \, [c(U^{(\ell)})]_k$, 
so $\lVert \Theta^{(\ell)\top} \diag(|\lambda^{(\ell)}|) \Theta^{(\ell)} \rVert_2 \le \sqrt{\gamma}\sup|{\sigma^{(\ell)}}''| \, \lVert \Theta^{(\ell)\top} \diag(c(U^{(\ell)})) \Theta^{(\ell)} \rVert_2$. Finally, using $\big \lVert \frac{\partial h^{(\ell-1)}}{\partial x} \big \rVert_2 \le \sqrt{\gamma}$, we obtain the desired result.
\end{proof}
}

{\color{revcolor}
\begin{theorem}\label{thm:lip_nn_third_order_bound}
Consider the LNN $F(x)$ defined in \eqref{eq:lip_nn}. If $\sigma^{(\ell)}$ is $\mathcal{C}^3$ for $\ell = 1, \dots, L-1$, then the spectral norm of $\frac{\partial^3 F_i(x)}{\partial x^2 \partial x_j}$ satisfies
\begin{align}\label{eq:lip_nn_third_order_bound}
\left \lVert \frac{\partial^3 F_i(x)}{\partial x^2 \partial x_j} \right \rVert_2 &\le \sum_{\ell=1}^{L-1} \Bigl[ \gamma\, \beta^{(\ell)} \sup|{\sigma^{(\ell)}}''| \, \Bigl( \Bigl\lVert \frac{\partial^2 F_i(x)}{\partial x_j \partial h^{(\ell)}} \Bigr\rVert_2 \notag \\
&+ 2 \Bigl\lVert \frac{\partial^2 h^{(\ell-1)}}{\partial x_j \partial x} \Bigr\rVert_2 \Bigr) 
+ \gamma^2 \sup|{\sigma^{(\ell)}}'''| \, \tau^{(\ell)} \Bigr],
\end{align}
where $\beta^{(\ell)} = \lVert \diag(c(U^{(\ell)}))^{1/2} W^{(\ell)} \rVert_2^2$, $\tau^{(\ell)} = \lVert {W^{(\ell)}}^\top \diag(c(U^{(\ell)}) \odot r(W^{(\ell)})) W^{(\ell)} \rVert_2 $, $[c(U^{(\ell)})]_k = \lVert U^{(\ell)}_{:,k} \rVert_2$, $[r(W^{(\ell)})]_k = \lVert W^{(\ell)}_{k,:} \rVert_2$, and the inner recursive Hessian bounds are given by:
\begin{equation}\label{eq:fwd_bound}
\left \lVert \frac{\partial^2 h^{(\ell)}}{\partial x_j \partial x} \right \rVert_2 \leq \sum_{k=1}^\ell \gamma \lVert c(U^{(k)}) \rVert_\infty \lVert W^{(k)} \rVert_2^2 \sup |{\sigma^{(k)}}'' |,
\end{equation}
\begin{equation}\label{eq:bwd_bound}
\left \lVert \frac{\partial^2 F_i(x)}{\partial x_j \partial h^{(\ell)}} \right \rVert_2 \leq \sum_{k=\ell+1}^{L-1} \gamma \lVert c(U^{(k)}) \rVert_\infty \lVert W^{(k)} \rVert_2^2 \sup |{\sigma^{(k)}}'' |.
\end{equation}
\end{theorem}
\begin{proof}
    The proof is given in Appendix~\ref{appx:proof_lip_nn_third_order_bound}.
\end{proof}
}

\begin{remark}
Many common nonlinear activation functions satisfy the conditions of Theorems~\ref{thm:lip_nn_hess_bound} and \ref{thm:lip_nn_third_order_bound}. Some classic examples include: (1) Softplus: $\sigma(x) = \log(1+e^x)$, $ 0 \leq \sigma'(x) \leq 1$, $0 \leq \sigma''(x) \leq 0.25$, $-0.097 \leq \sigma'''(x) \leq 0.097$; (2) Sigmoid: $\sigma(x) = \frac{1}{1+e^{-x}}$, $0 \leq \sigma'(x) \leq 0.25$, $-0.097 \leq \sigma''(x) \leq 0.097$, $-1/8 \leq \sigma'''(x) \leq 0.042$; (3) Hyperbolic tangent: $\sigma(x) = \tanh(x)$, $0 \leq \sigma'(x) \leq 1$, $-0.77 \leq \sigma''(x) \leq 0.77$, $-2 \leq \sigma'''(x) \leq 2/3$.
\end{remark}

{\color{revcolor}
\begin{remark}\label{rmk:asymptotic_bound}
We characterize the asymptotic growth of these bounds in $L$. For a Sandwich Layer \cite{wang2023direct} with $\Psi^{(\ell)} = \kappa I$ ($\kappa > 0$), we have $\lVert U^{(\ell)} \rVert_2 \leq \sqrt{2}\kappa$ and $\lVert W^{(\ell)} \rVert_2 \leq \sqrt{2}/\kappa$; with identical activations $\sigma$, Theorems~\ref{thm:lip_nn_hess_bound} and~\ref{thm:lip_nn_third_order_bound} simplify to $\big \lVert \frac{\partial^2 F_i(x)}{\partial x^2} \big \rVert_2 \leq \frac{1}{\kappa}(2 \gamma)^{\frac{3}{2}} \sup |{\sigma}''| \, (L-1)$ and $\big \lVert \frac{\partial^3 F_i(x)}{\partial x^2 \partial x_j} \big \rVert_2 \leq \frac{4 \gamma^2}{\kappa^2} \bigl[3 ( \sup |{\sigma}''| )^2 (L-2) (L-1) + \sup |{\sigma}'''| \, (L-1) \bigr]$, i.e., linear and quadratic in $L$, respectively. The same rates hold for Spectral Normalization \cite{miyato2018spectral} and Orthogonal Layers \cite{trockman2021orthogonalizing}, with $\big\lVert \frac{\partial^2 F_i(x)}{\partial x^2} \big\rVert_2 \leq \gamma^{\frac{3}{2}} \sup|\sigma''| \, (L-1)$ and $\big\lVert \frac{\partial^3 F_i(x)}{\partial x^2 \partial x_j} \big\rVert_2 \leq \gamma^2 [ \tfrac{3}{2} (\sup|\sigma''|)^2 (L-1)(L-2) + \sup|\sigma'''| \, (L-1) ]$.

Figure~\ref{fig:hess_bound} compares this Hessian bound of LNNs against that of standard fully connected networks (FCNs) computed using the result in \cite[Theorem~4]{singla2020second}. We consider the setup $n_x = 2$, $n_{L}=1$, $n_{\ell}=16$ (for $\ell = 1,\dots,L-1$), and vary $L$ from $2$ up to $5$, using $\tanh(\cdot)$ as the activation function. For each value of $L$, we randomly initialize the network parameters 20 times and plot the mean and standard deviation of the resulting Hessian bounds for both architectures (where the LNNs use Orthogonal Layers \cite{trockman2021orthogonalizing}). We observe that the unconstrained FCN bounds grow much faster than those of the LNNs. Figure~\ref{fig:third_order_bound_lip} visualizes the third-order bound of these randomly initialized LNNs. Furthermore, we note that Theorems~\ref{thm:lip_nn_hess_bound} and \ref{thm:lip_nn_third_order_bound} present global, closed-form, layer-separable upper bounds, and the empirical bars merely lower-bound this global supremum. Although strict equality is only attained in degenerate worst-case networks, this theoretical looseness is reclaimed during the BnB process in Section~\ref{sec:bnb}: the bounding constants appear only in higher-order remainders that vanish rapidly as the search rectangles shrink.
\end{remark}
}

\begin{figure}
\centering
\begin{subfigure}[b]{0.23\textwidth}
    \centering
    \includegraphics[width=\textwidth]{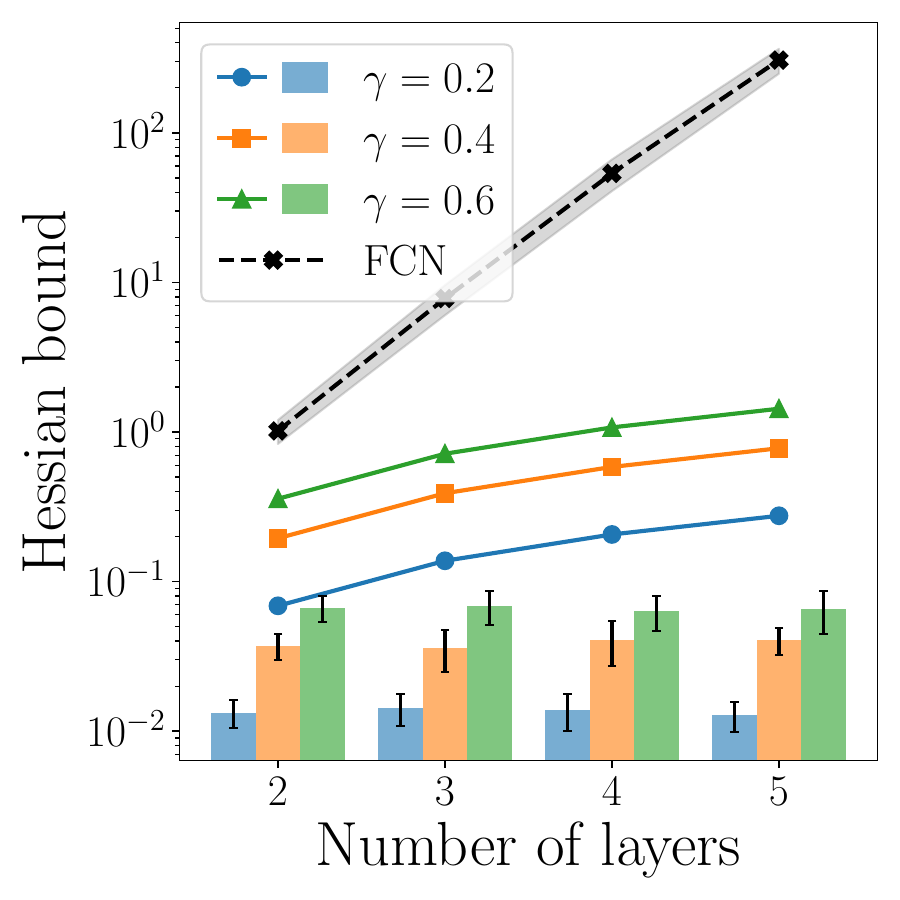}
    \caption{LNN and FCN Hessian bounds.}
    \label{fig:hess_bound}
\end{subfigure}
\begin{subfigure}[b]{0.23\textwidth}
    \centering
    \includegraphics[width=\textwidth]{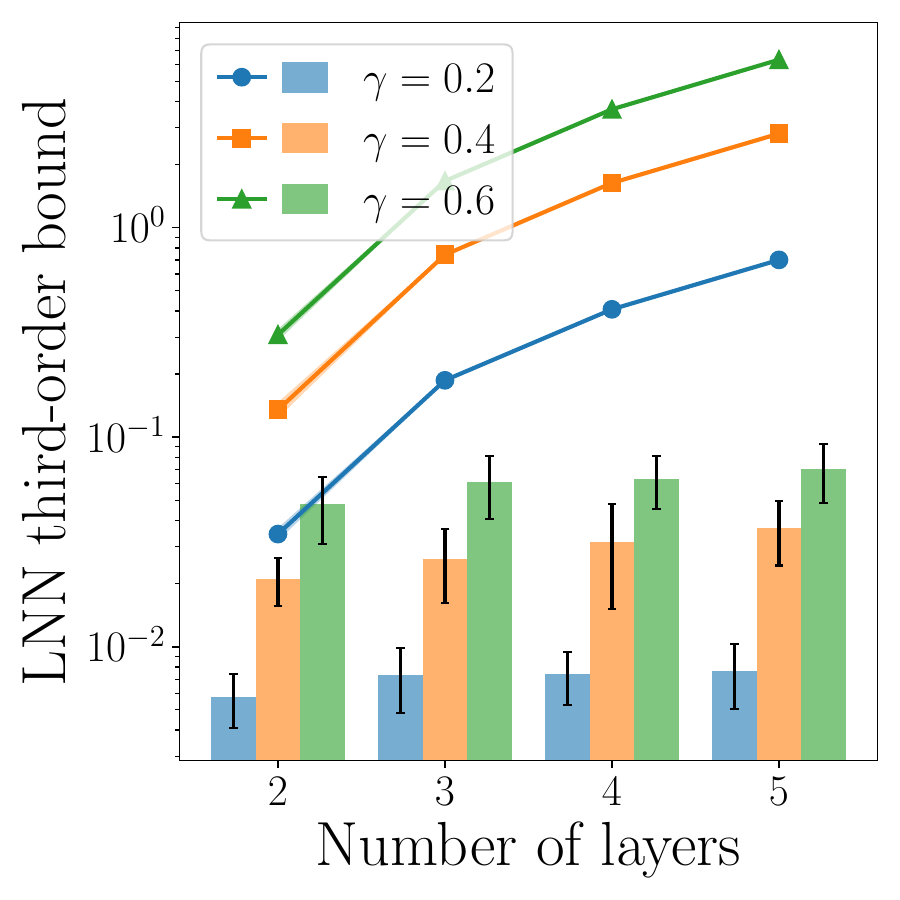}
    \caption{LNN third-order bound.}
    \label{fig:third_order_bound_lip}
\end{subfigure}
\caption{\subref{fig:hess_bound}: Comparison of Hessian bounds for FCNs and LNNs as the number of layers increases. Theoretical bounds are denoted by solid lines, while empirical bounds are shown as bars. \subref{fig:third_order_bound_lip}: Third-order bound of LNNs w.r.t. the number of layers.}
\label{fig:hess_and_third_order_bound}
\end{figure}

\begin{remark}
Beyond Lyapunov verification, Hessian bounds have also been applied in robustness certification \cite{singla2020second}, adversarial training \cite{srinivas2022efficient}, and neural certificates for dynamical systems \cite{neustroev2025neural}, underscoring their broader role as a fundamental tool for tighter characterizations of neural network behavior.
\end{remark}

\subsection{Branch-and-Bound Scheme}\label{sec:bnb}
Verification has been the main challenge in training neural network Lyapunov functions. Although positive definite, Lyapunov functions are not necessarily convex \cite{jongeneel2024continuation}. Existing neural network verification approaches include interval methods \cite{hein2017formal, wang2018efficient}, convex relaxations \cite{wong2018provable, fazlyab2020safety}, SMT solvers \cite{gao2013dreal}, etc. Interval methods often give vacuous bounds on large domains, convex relaxations focus on input–output robustness rather than Lyapunov conditions, and SMT solvers face scalability limits. {\color{revcolor} We propose a novel high-order verification scheme. The core idea is to use the smoothness of the network and explicitly construct and incorporate analytical first- and second-order bounds through Taylor expansions. These high-order bounds create significantly tighter bounds of the network output compared with zeroth-order bounds, which in turn drastically prunes the search space at the cost of slightly more expensive per-node evaluations. }

{\color{revcolor}
Let $Q = [x_{\mathrm{lb}}, x_{\mathrm{ub}}] \subset \R^{n_x}$ be a hyperrectangle. We seek to verify that a target function $\Phi \in \mathcal{C}^3$ satisfies $\Phi(x) \geq 0$ to a prescribed accuracy $\delta > 0$ on $Q \backslash E$ where $E \subset \R^{n_x}$ is an excluded set. We extract the midpoint $x_m (Q) = \frac{1}{2}(x_{\mathrm{ub}} + x_{\mathrm{lb}})$ and the half-width vector $v(Q) = \frac{1}{2}(x_\mathrm{ub} - x_\mathrm{lb} )$ of $Q$, and we omit the dependency of $x_m$ and $v$ on $Q$ for brevity. The goal is to find lower and upper bound functions $\Phi_{\mathrm{lb}}$ and $\Phi_{\mathrm{ub}}$ satisfying $\Phi_{\mathrm{lb}} (Q) \leq \min_{x \in Q} \Phi(x) \leq \Phi_{\mathrm{ub}} (Q)$ for a given $Q$. In the following, we propose three options for $\Phi_{\mathrm{lb}}$ and $\Phi_{\mathrm{ub}}$.}

{\color{revcolor}
\textbf{Zeroth-order bounds:} Let $K_\Phi \geq \Lip (\Phi)$. By the definition of Lipschitz continuity, we have:
\begin{subequations}\label{eq:zeroth_order_bound}
\begin{align}
    \Phi_{\mathrm{lb}}^{(0)}(Q) &= \Phi(x_m) - K_\Phi \lVert v \rVert_2,\\
    \Phi_{\mathrm{ub}}^{(0)}(Q) &= \Phi(x_m).
\end{align}
\end{subequations}

\textbf{First-order bounds:} Because $\Phi \in \mathcal{C}^2$, applying the second-order Taylor expansion with the Lagrange remainder yields: $\Phi(x) = \Phi(x_m) + \frac{\partial \Phi}{\partial x}(x_m) (x-x_m) + \frac{1}{2} (x-x_m)^\top \frac{\partial^2 \Phi}{\partial x^2} (x') (x-x_m)$ where $x' = x_m + \lambda (x-x_m)$ for some $\lambda \in (0, 1)$. Defining the spectral bound $H_\Phi \geq \sup_{x \in \R^{n_x}} \bigl \lVert \frac{\partial^2 \Phi}{\partial x^2} (x) \bigr \rVert_2$, we establish tighter first-order bounds:

\begin{subequations}\label{eq:first_order_bound}
    \begin{align}
    \Phi_{\mathrm{lb}}^{(1)}(Q) &= \Phi(x_m) - \left \lvert \frac{\partial \Phi}{\partial x}(x_m) \right \rvert v - \frac{1}{2} H_\Phi \lVert v \rVert_2^2, \\
    \Phi_{\mathrm{ub}}^{(1)}(Q) &= \Phi(x_m) - \left \lvert \frac{\partial \Phi}{\partial x}(x_m) \right \rvert v + \frac{1}{2} H_\Phi \lVert v \rVert_2^2.
    \end{align}
\end{subequations}

\textbf{Second-order bounds:} Similarly, as $\Phi \in \mathcal{C}^3$, we define $T_\Phi \geq \sup_{x \in \R^{n_x}} \left \lVert \frac{\partial^3 \Phi}{\partial x^2 \partial x_j} (x) \right \rVert_2$ for $j = 1, ..., n_x$. By the third-order Taylor expansion, we get more accurate bounds:
\begin{subequations}\label{eq:second_order_bound}
    \begin{align}
    \Phi_{\mathrm{lb}}^{(2)}(Q) &= \Phi(x_m) - \left \lvert \frac{\partial \Phi}{\partial x}(x_m) \right \rvert v - \frac{1}{2} v^\top \left \lvert \frac{\partial^2 \Phi}{\partial x^2}(x_m) \right \rvert v \notag \\
    &\quad -\frac{1}{6} T_\Phi \lVert v \rVert_2^2 \lVert v \rVert_1, \\
    \Phi_{\mathrm{ub}}^{(2)}(Q) &= \Phi(x_m) - \left \lvert \frac{\partial \Phi}{\partial x}(x_m) \right \rvert v + \frac{1}{2} v^\top \left \lvert \frac{\partial^2 \Phi}{\partial x^2}(x_m) \right \rvert v \notag \\
    &\quad + \frac{1}{6} T_\Phi \lVert v \rVert_2^2 \lVert v \rVert_1 .
    \end{align}
\end{subequations}

To obtain the tightest combined enclosure subject to the smoothness of $\Phi$, we take $\Phi_{\mathrm{lb}} (Q) = \max(\Phi_{\mathrm{lb}}^{(0)}, \Phi_{\mathrm{lb}}^{(1)}, \Phi_{\mathrm{lb}}^{(2)})$ and $\Phi_{\mathrm{ub}} (Q) = \min(\Phi_{\mathrm{ub}}^{(0)}, \Phi_{\mathrm{ub}}^{(1)}, \Phi_{\mathrm{ub}}^{(2)})$.} The detailed BnB algorithm is given in Algorithm~\ref{alg:bnb}. The algorithm starts from the initial rectangle $Q_{\mathrm{init}}$ and keeps splitting the rectangles into smaller ones until one of the early stop conditions is met or the required precision is achieved. 

\begin{algorithm}[t]
\caption{High-Order Branch-and-Bound Verification}
\label{alg:bnb}
{\color{revcolor}
\begin{algorithmic}[1]
\REQUIRE Target function $\Phi$, bounds $\Phi_{\mathrm{lb}}$ and $\Phi_{\mathrm{ub}}$,  excluded region $E$, initial rectangle $Q_{\text{init}} \not\subset E$, tolerance $\delta$, maximum depth $k_{\text{max}}$;

\STATE Initialize $k=0$ and $\mathcal{Q}_0 = \{ Q_{\text{init}} \}$;

\LOOP

\STATE Initialize $L_k = \infty$, $U_k = \infty$, and $\mathcal{Q}_{k+1} = \varnothing$;

\FOR{$Q \in \mathcal{Q}_k$} \label{lst:line:start}
    \STATE $L_k = \min (L_k, \Phi_{\mathrm{lb}} (Q))$; \label{lst:line:def_L}
    \STATE $U_k = \min (U_k, \Phi_{\mathrm{ub}} (Q))$ if $Q \cap E = \varnothing$; \label{lst:line:def_U}

    \IF{$\Phi_{\mathrm{lb}} (Q) < 0$}
    \STATE Split $Q$ along its longest axis into $Q_1$ and $Q_2$;

    \STATE $\mathcal{Q}_{k+1} = \mathcal{Q}_{k+1} \cup \{Q_i\}$ if $Q_i \not\subset E$ for $i \in \{1, 2\}$; \label{lst:line:Q_subset_E}

    \ENDIF
\ENDFOR \label{lst:line:end}

\STATE If $U_k < 0$, return \textbf{failure} (early stop);

\STATE If $L_k \geq 0$, return \textbf{success} (early stop);

\STATE If $U_k - L_k \leq  \delta$, return \textbf{success};

\STATE If $k \geq k_{\text{max}}$, return \textbf{unknown} (early stop);

\STATE $k = k + 1$;

\ENDLOOP

\end{algorithmic}}
\end{algorithm}

Finally, we evaluate the worst-case convergence of this scheme. Given $Q$, its volume and condition number are:
\begin{align}
    \vol(Q) &= \prod_{i=1}^{n_x} (x_{\mathrm{ub},i} - x_{\mathrm{lb},i}), \label{eq:volume_Q} \\
    \cond(Q) &= \frac{\max_{i=1, ..., n_x} (x_{\mathrm{ub},i} - x_{\mathrm{lb},i})}{\min_{i=1, ..., n_x} (x_{\mathrm{ub},i} - x_{\mathrm{lb},i})}. \label{eq:cond_Q}
\end{align}

\begin{lemma}[Lemma 1 in \cite{balakrishnan1991branch}] \label{lemma:cond}
For any $k$ and any rectangle $Q \in \mathcal{Q}_k$ generated in Algorithm~\ref{alg:bnb}, we have 
\begin{equation}
\cond(Q) \leq \max \{ \cond(Q_{\mathrm{init}}), 2 \}.
\end{equation}
\end{lemma}

{\color{revcolor}
\begin{theorem}\label{thm:bnb_complexity_zeroth_order}
Assume that $\Phi$ is Lipschitz continuous and $K_\Phi \geq \Lip(\Phi)$. If none of the early stop conditions are satisfied and $E = \varnothing$, then Algorithm~\ref{alg:bnb} finishes within $k_0 = \lceil \log_2 (P_0) \rceil$ iterations, evaluating at most $\mathcal{O}(P_0)$ rectangles, where $P_0 = \bigl(\frac{\sqrt{n_x} K_{\Phi} \max \{ \cond(Q_{\mathrm{init}}), 2 \}}{2 \delta}\bigr)^{n_x} \vol (Q_{\mathrm{init}})$.
\end{theorem}

\begin{proof}
By the equivalence of norms, $\lVert v (Q) \rVert_2 \leq \frac{\sqrt{n_x}}{2} \max_i (x_{\mathrm{ub},i} - x_{\mathrm{lb},i})$. Because $\cond(Q) \geq 1$, we lower bound the volume by $\vol(Q) \geq \max_i (x_{\mathrm{ub},i} - x_{\mathrm{lb},i}) [ \min_i (x_{\mathrm{ub},i} - x_{\mathrm{lb},i})]^{n_x-1} = \frac{[ \max_i (x_{\mathrm{ub},i} - x_{\mathrm{lb},i}) ]^{n_x}}{\cond(Q)^{n_x-1}} \geq \bigl( \frac{\max_i (x_{\mathrm{ub},i} - x_{\mathrm{lb},i})}{\cond(Q)} \bigr)^{n_x} $. This implies $\vol(Q) \geq \bigl( \frac{ 2 \lVert v \rVert_2}{\sqrt{n_x} \cond(Q)} \bigr)^{n_x}$. Rearranging provides a bound on the half-width norm of any $Q$:
\begin{equation}\label{eq:bound_size_Q}
\lVert v (Q) \rVert_2 \leq \frac{\sqrt{n_x}}{2} \cond(Q) \vol(Q)^{\frac{1}{n_x}}.
\end{equation}

Because the algorithm always splits a rectangle into two, the volume at iteration $k$ is $\vol(Q) = \frac{1}{2^k} \vol (Q_{\mathrm{init}})$. For any $Q \in \mathcal{Q}_k$, applying this and Lemma~\ref{lemma:cond} yields:
\begin{equation}\label{eq:bound_size_Q_k}
\lVert v (Q) \rVert_2 \leq \frac{\sqrt{n_x}}{2} \max \{ \cond(Q_{\mathrm{init}}), 2 \} \left ( \frac{\vol (Q_{\mathrm{init}})}{2^k} \right )^{\frac{1}{n_x}}.
\end{equation} 

Using the zeroth-order bounds \eqref{eq:zeroth_order_bound}, we have $L_k = \min_{Q \in \mathcal{Q}_k} [\Phi(x_m (Q)) - K_{\Phi} \lVert v (Q) \rVert_2 ] \geq U_k - \allowbreak K_{\Phi} \max_{Q \in \mathcal{Q}_k} \lVert v (Q) \rVert_2$. Therefore, by \eqref{eq:bound_size_Q_k}, $U_k - L_k \leq \frac{\sqrt{n_x}}{2} K_{\Phi}  \max \{ \cond(Q_{\mathrm{init}}), 2 \} \bigl ( \frac{\vol (Q_{\mathrm{init}})}{2^k} \bigr )^{\frac{1}{n_x}}$. To obtain $U_k - L_k \leq  \delta $ for any $\delta > 0$, we need $k_0 = \lceil \log_2 (P_0) \rceil$ iterations where 
\begin{equation}
P_0 = \left (\frac{\sqrt{n_x} K_{\Phi} \max \{ \cond(Q_{\mathrm{init}}), 2 \}}{2 \delta}\right )^{n_x} \vol (Q_{\mathrm{init}}).
\end{equation}
After $k_0$ iterations, the total number of evaluated rectangles is $2^{k_0 + 1} - 1$ which is $\mathcal{O}(P_0)$.
\end{proof}

\begin{theorem}\label{thm:bnb_complexity_first_order}
Under the same conditions as Theorem~\ref{thm:bnb_complexity_zeroth_order}, if $\Phi \in \mathcal{C}^2$ and $H_\Phi \geq \sup_{x \in \R^{n_x}} \bigl \lVert \frac{\partial^2 \Phi}{\partial x^2} (x) \bigr \rVert_2$, then Algorithm~\ref{alg:bnb} finishes within $k_1 = \lceil \log_2 (P_1) \rceil$ iterations, evaluating at most $\mathcal{O}(P_1)$ rectangles, where $P_1 =  \bigl(\frac{ n_x H_{\Phi} \max^2 \{ \cond(Q_{\mathrm{init}}), 2 \}}{4 \delta}\bigr)^{\frac{n_x}{2}} \vol (Q_{\mathrm{init}})$.
\end{theorem}
\begin{proof}
By the definition of $U_k$ and $L_k$ and the first-order bounds \eqref{eq:first_order_bound}, we have $U_k \leq \min_{Q \in \mathcal{Q}_k} [ \Phi(x_m) - \left \lvert \frac{\partial \Phi}{\partial x}(x_m) \right \rvert v (Q) + \frac{1}{2} H_\Phi \lVert v (Q) \rVert_2^2 ] \leq L_k + H_\Phi \max_{Q \in \mathcal{Q}_k} \lVert v (Q) \rVert_2^2$. The rest of the proof follows the algebraic rearrangement established in Theorem~\ref{thm:bnb_complexity_zeroth_order}.
\end{proof}

When $\Phi(x) = V_{\theta_1} (x)$, $n_x = 4$, $\sigma = \tanh$, $\cond(Q_{\mathrm{init}}) = 1$, $\vol(Q_{\mathrm{init}}) = 1$, and $\delta = 0.001$, the resulting BnB depth of Theorem~\ref{thm:bnb_complexity_first_order} is shown in Fig.~\ref{fig:bnb_depth} for Sandwich Layers \cite{wang2023direct}.

\begin{remark}
The lines~\ref{lst:line:start}-\ref{lst:line:end} of Algorithm~\ref{alg:bnb} can be accelerated using parallel computation on a GPU. As demonstrated in Section~\ref{sec:simulation_studies}, this approach not only offers runtime benefits but also scales better to larger neural networks.
\end{remark}
}

\begin{figure}[t]
    \centering
    \includegraphics[width=0.7\linewidth]{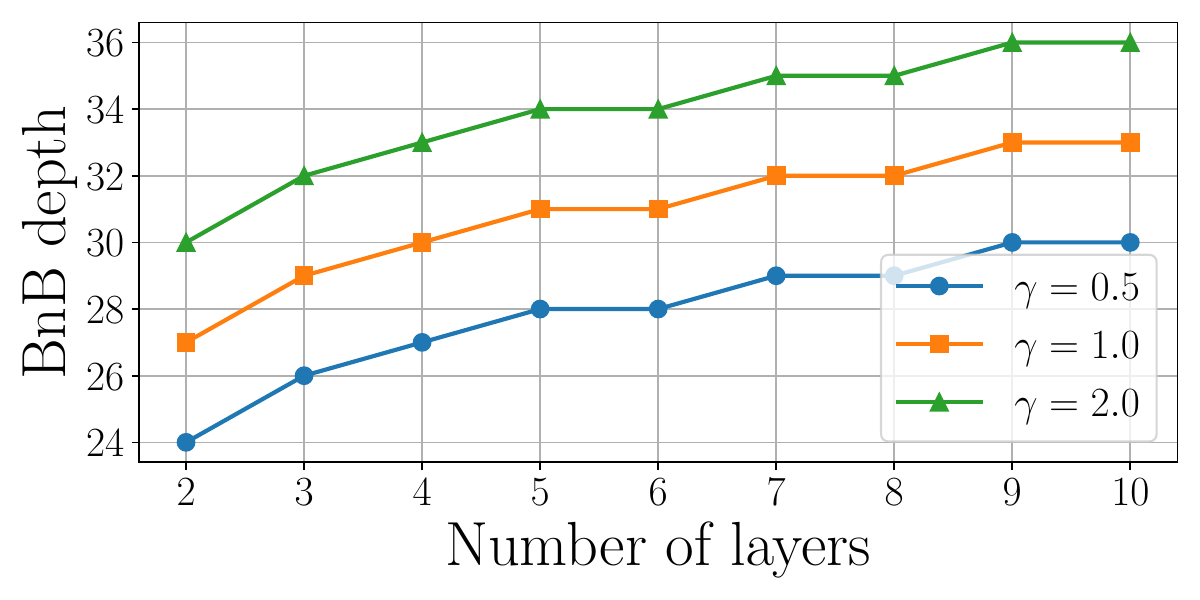}
    \caption{BnB depth w.r.t. the number of layers and Lipschitz bound.}
    \label{fig:bnb_depth}
\end{figure}

\subsection{Verification of the RCLF}
In this subsection, we detail the application of the high-order BnB algorithm (Algorithm~\ref{alg:bnb}) to formally verify the three core conditions of an RCLF.

\subsubsection{Positive Definiteness} 
To verify condition \eqref{eq:pd_cond}, we must guarantee that $V_{\theta_1}(x) > 0$ strictly away from the origin. To avoid numerical intractability near the equilibrium \cite{chang2019neural, zhou2022neural, wu2023neural}, we verify this property over the domain $\mathcal{X} \setminus B_2(0,\eta)$ for a small margin $\eta \ll 1$. Thus, we set the target function $\Phi(x) = V_{\theta_1}(x)$ and the excluded region $E = B_2(0,\eta)$.

Because $V_{\theta_1}$ is parameterized to be $\gamma_V$-Lipschitz, we set the zeroth-order bound constant to $K_\Phi = \gamma_V$. If the chosen activation functions are $\mathcal{C}^3$, the spectral bounds from Theorems~\ref{thm:lip_nn_hess_bound} and \ref{thm:lip_nn_third_order_bound} immediately provide the constants for the higher-order bounds in \eqref{eq:first_order_bound} and \eqref{eq:second_order_bound}. 

During the pruning step (line~\ref{lst:line:Q_subset_E} in Algorithm~\ref{alg:bnb}), a rectangle $Q = [x_{\mathrm{lb}}, x_{\mathrm{ub}}]$ is fully contained within the excluded origin ball ($Q \subset E$) if $ \sum_{i=1}^{n_x} \max(|x_{\mathrm{lb},i}|, |x_{\mathrm{ub},i}|)^2 \leq \eta^2 $.

\subsubsection{Inclusion of $B_2(0, \mu)$ by $\Omega_{\widehat{V}_r}$} 
Definition~\ref{def:rclf} requires that the ball $B_2(0, \mu)$ is entirely contained within the verified sublevel set $\Omega_{\widehat{V}_r}= \{ x \in \mathcal{X} \mid V_{\theta_1}(x) \leq \widehat{V}_r \}$. Because $V_{\theta_1}$ is $\gamma_V$-Lipschitz, this condition is trivially satisfied if $\widehat{V}_r \geq \gamma_V \mu$. 

If this shortcut does not hold, we cast the inclusion check as a verification problem by setting $\Phi(x) = \widehat{V}_r - V_{\theta_1}(x)$, initializing the root node as $Q_{\mathrm{init}} = [-\mu , \mu]^{n_x}$, and defining the excluded region as $E = Q_{\mathrm{init}} \setminus B_2(0, \mu)$. Algorithm~\ref{alg:bnb} then checks if $\Phi(x) \geq 0$ over $B_2(0, \mu)$. As before, high-order bounds apply if $V_{\theta_1} \in \mathcal{C}^3$. 

Here, the pruning condition $Q \subset E$ is met if the rectangle lies strictly outside the ball $B_2(0, \mu)$. This occurs when its minimum squared distance to the origin exceeds $\mu^2$, which is $\sum_{i=1}^{n_x} \max(0, x_{\mathrm{lb}, i}, -x_{\mathrm{ub}, i})^2 > \mu^2$.

\subsubsection{Decrease Condition} 
To verify the stability condition \eqref{eq:stability_cond}, we evaluate $H(x) + \omega(x) \leq 0$ (which is a stronger guarantee than \eqref{eq:stability_cond}, as derived in \eqref{eq:def_H}). We restrict our search space to the verified forward invariant set $\Omega_{\widehat{V}_r}$ and exclude the interior of $B_2(0, \mu)$. We therefore set $\Phi(x) = - H(x) - \omega(x)$ and $E = \mathrm{int} (B_2(0, \mu)) \cup \{ x \in \mathcal{X} \mid V_{\theta_1}(x) > \widehat{V}_r \}$. 

A rectangle $Q$ is pruned if it is completely inside $B_2(0, \mu)$ or strictly outside the sublevel set:
$ \sum_{i=1}^{n_x} \max(|x_{\mathrm{lb},i}|, |x_{\mathrm{ub},i}|)^2 < \mu^2 $ or $V_{{\theta_1}, \mathrm{lb}} (Q) > \widehat{V}_r $
where $V_{{\theta_1}, \mathrm{lb}}$ is the tightest available lower bound for $V_{\theta_1}$ on $Q$, obtained via \eqref{eq:zeroth_order_bound}, \eqref{eq:first_order_bound}, or \eqref{eq:second_order_bound} subject to the smoothness of $V_{\theta_1}$.

Due to input saturation and the norm operations within $H(x)$, the target function $\Phi(x)$ is only $\mathcal{C}^0$. Its global Lipschitz bound $K_\Phi$ can be computed by:
$$
\begin{aligned}
\Lip(H+\omega) &\leq H_V C_f + \gamma_V (C_{f_x} + \gamma_{\pi} C_{f_u} ) \\
&\quad + K_\omega + H_V \lVert G \rVert_2 C_\epsilon + \gamma_V \lVert G \rVert_2 K_\epsilon \defeq K_\Phi
\end{aligned}
$$
where $H_V \geq \sup_{x \in \R^{n_x}} \bigl \lVert \frac{\partial^2 V_{\theta_1}}{\partial x^2} (x) \bigr \rVert_2$, $K_\omega \geq \Lip(\omega)$, $C_f = \max_{x \in Q, u \in \mathcal{U}} \lVert f(x,u) \rVert_2$, $C_\epsilon = \max_{x \in Q} \epsilon(x)$, $C_{f_x} = \max_{x \in Q, u \in \mathcal{U}}  \lVert \frac{\partial f}{\partial x}(x,u) \rVert_2$, $C_{f_u} = \max_{x \in Q, u \in \mathcal{U}}  \lVert \frac{\partial f}{\partial u}(x,u) \rVert_2$, and $K_\epsilon \geq \Lip (\epsilon)$ due to the boundedness of $\mathcal{X}$ and $\mathcal{U}$. Because a direct Taylor expansion on the non-smooth $\Phi$ is impossible, we derive a custom high-order lower bound for $\min_{x \in Q} \Phi(x)$ in Appendix~\ref{appx:high_order_bound_stability}.

\begin{remark}
Our framework can be extended to address parametric uncertainties. For $\dot{x} = f(x,u,\lambda)$ with the uncertain parameters $\lambda \in \Lambda \subset \mathbb{R}^{n_\lambda}$ and state-feedback controller $\pi_{\theta_2}$, the closed-loop system is $\dot{x} = f_{\mathrm{cl}} (x, \lambda) = f(x,\pi_{\theta_2}(x),\lambda)$. We enforce a uniform Lyapunov decrease in the loss function: $\sup_{\lambda \in \Lambda}  \frac{\partial V_{\theta_1}}{\partial x}(x) f_{\mathrm{cl}} (x, \lambda) + \omega(x) \leq 0$. The $\sup_{\lambda \in \Lambda}$ operator can be implemented via adversarial sampling of $\lambda$, and/or exploiting structure: if $f$ is affine in $\lambda$ and $\Lambda$ is a convex polytope, the supremum is attained at vertices. Verification is analogous: check the Lyapunov inequality at the vertices of $\Lambda$ if the structure allows; otherwise, apply BnB over $(x, \lambda)$.
\end{remark}

{\color{revcolor}
\begin{remark}\label{rmk:computational_tradeoff}
Theorems~\ref{thm:bnb_complexity_zeroth_order} and \ref{thm:bnb_complexity_first_order} bound the number of evaluated rectangles, so the overall time complexity is the BnB tree size times the per-node bound computation cost. The global constants $H_\Phi$ and $T_\Phi$ are computed once offline using Theorems~\ref{thm:lip_nn_hess_bound} and \ref{thm:lip_nn_third_order_bound}, so the per-node cost is dominated by the local derivatives $\frac{\partial \Phi}{\partial x}(x_m)$ and $\frac{\partial^2 \Phi}{\partial x^2}(x_m)$ (e.g., $\Phi = V_{\theta_1}$ for positive definiteness and $\Phi = \widehat{V}_r - V_{\theta_1}$ for inclusion). With analytical forward-mode differentiation, this cost is $\mathcal{O}(L_V n_V^2)$, $\mathcal{O}(L_V n_V^2 n_x)$, and $\mathcal{O}(L_V n_V^2 n_x^2)$ for the zeroth-, first-, and second-order bounds, respectively, where $L_V$ and $n_V$ are the number of layers and maximum layer width of $V_{\theta_1}$. Although incorporating higher-order derivatives increases the per-node computational cost polynomially in $n_x$, it produces significantly tighter bounds. This trade-off drastically reduces the overall BnB depth, preventing the explosion of the search tree that empirically causes the zeroth-order baseline to become computationally intractable (see Appendix~\ref{appx:comparison_between_bounds}).
\end{remark}
}

{\color{revcolor} 
\begin{remark}\label{rmk:small_lipschitz_bound} 
Although the RCLF can be rescaled (Property~\ref{prpt:scaling}), a much smaller Lipschitz bound is undesirable in practice: at a fixed absolute tolerance $\delta$, it shrinks the bounding error and the signal together, so the apparent speedup is merely degraded relative precision; and it flattens $V_{\theta_1}$, yielding vanishing, ill-conditioned gradients that make the network harder to train.
\end{remark}
}

{\color{revcolor}
\begin{remark}\label{rmk:bi_lipschitz}
Parameterizing $\widetilde{V}(x) = \frac{1}{2} \lVert g(x) \rVert_2^2$ with a bi-Lipschitz network $g$ \cite{cheng2024learning, wang2024monotone} intrinsically guarantees positive definiteness and topologically spherical level sets, eliminating both the positive definiteness loss term and its verification step. This is a promising direction, but deriving the high-order bounds for this composition poses additional analytical challenges and is left to future work.
\end{remark}
}

{\color{revcolor}
\begin{remark}\label{rmk:sphere_equiv}
Although Lyapunov functions for asymptotically stable systems have level sets topologically equivalent to spheres \cite{kvalheim2025global}, our framework targets systems with persistent disturbances and achieves UUB, for which it suffices to verify the forward invariance of $\Omega_{\widehat{V}_r}$ to trap trajectories.
\end{remark}
}

\begin{figure*}[t]
    \centering
    \begin{subfigure}[b]{0.18\textwidth}
        \centering
        \includegraphics[width=\textwidth]{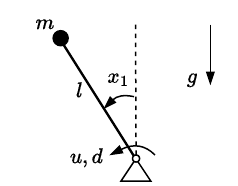}
        \caption{Inverted pendulum}
        \label{fig:inverted_pendulum_diagram}
    \end{subfigure}
    \begin{subfigure}[b]{0.18\textwidth}
        \centering
        \includegraphics[width=\textwidth]{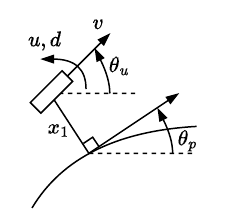}
        \caption{Unicycle}
        \label{fig:unicycle_diagram}
    \end{subfigure}
    \begin{subfigure}[b]{0.18\textwidth}
        \centering
        \includegraphics[width=\textwidth]{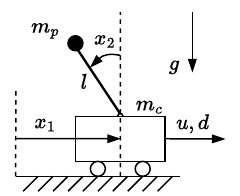}
        \caption{Cartpole}
        \label{fig:cartpole_diagram}
    \end{subfigure}
    \begin{subfigure}[b]{0.18\textwidth}
        \centering
        \includegraphics[width=\textwidth]{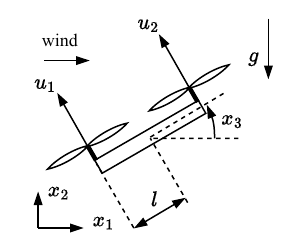}
        \caption{2D quadrotor}
        \label{fig:quadrotor_2d_diagram}
    \end{subfigure}
    \begin{subfigure}[b]{0.18\textwidth}
        \centering
        \includegraphics[width=\textwidth]{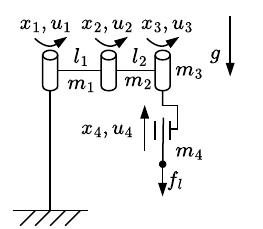}
        \caption{SCARA arm}
        \label{fig:scara_diagram}
    \end{subfigure}
    \caption{Diagrams of the inverted pendulum, unicycle path following, cartpole, 2D quadrotor, and SCARA arm.}
    \label{fig:system_diagrams}
\end{figure*}

\begin{table*}[t]
\centering
\caption{Verification time (in seconds unless otherwise specified) of the positive definiteness, decrease, and inclusion (of $B_2(0,\mu)$ by $\Omega_{\widehat{V}_r}$) conditions for different systems and different numbers of layers in $V_{\theta_1}$ (denoted by $L_V$).}
\label{tab:time}
\begin{threeparttable}
\scriptsize
\setlength{\tabcolsep}{2.5pt}
\begin{tabular}{C{1.4cm} C{0.35cm} *{12}{C{1.12cm}}}
\toprule
\multirow{3}{*}{\makecell[c]{System}} & \multirow{3}{*}{$L_V$}
 & \multicolumn{4}{c}{Positive Definiteness}
 & \multicolumn{4}{c}{Decrease}
 & \multicolumn{4}{c}{Inclusion} \\
\cmidrule(lr){3-6}\cmidrule(lr){7-10}\cmidrule(lr){11-14}
 & & \multirow{2}{*}{dReal} & \multirow{2}{*}{\makecell{Ours\\(CPU)}} & \multirow{2}{*}{\makecell{$\alpha,\beta$-\\CROWN\tnote{a}}} & \multirow{2}{*}{\makecell{Ours\\(GPU)}}
   & \multirow{2}{*}{dReal} & \multirow{2}{*}{\makecell{Ours\\(CPU)}} & \multirow{2}{*}{\makecell{$\alpha,\beta$-\\CROWN\tnote{a}}} & \multirow{2}{*}{\makecell{Ours\\(GPU)}}
   & \multirow{2}{*}{dReal} & \multirow{2}{*}{\makecell{Ours\\(CPU)}} & \multirow{2}{*}{\makecell{$\alpha,\beta$-\\CROWN\tnote{a}}} & \multirow{2}{*}{\makecell{Ours\\(GPU)}} \\
& & & & & & & & & & & & & \\
\midrule

\multirow{3}{*}{\makecell[c]{Inverted\\Pendulum\\($n_x=2$)}}
& 2 & 0.06{\tiny$\pm$0.007} & \textbf{0.03{\tiny$\pm$0.006}} & 0.93{\tiny$\pm$0.03} & 0.31{\tiny$\pm$0.02}
    & $\approx$ 989\tnote{b} & 0.79{\tiny$\pm$0.003} & 4.39{\tiny$\pm$0.02} & \textbf{0.71{\tiny$\pm$0.02}}
    & -\tnote{c} & - & - & - \\
& 3 & 55.16{\tiny$\pm$10.92} & \textbf{0.04{\tiny$\pm$0.009}} & 1.27{\tiny$\pm$0.04} & 0.30{\tiny$\pm$0.02}
    & $>$ 1~\si{h} & 0.72{\tiny$\pm$0.03} & 6.67{\tiny$\pm$0.03} & \textbf{0.71{\tiny$\pm$0.01}}
    & - & - & - & - \\
& 4 & $>$ 1~\si{h} & \textbf{0.05{\tiny$\pm$0.04}} & 1.67{\tiny$\pm$0.07} & 0.33{\tiny$\pm$0.03}
    & $>$ 1~\si{h} & 0.74{\tiny$\pm$0.01} & 8.75{\tiny$\pm$0.13} & \textbf{0.62{\tiny$\pm$0.01}}
    & - & - & - & - \\
\midrule

\multirow{3}{*}{\makecell[c]{Unicycle Path\\Following\\($n_x=2$)}}
& 2 & 0.10{\tiny$\pm$0.003} & \textbf{0.02{\tiny$\pm$0.0004}} & 0.96{\tiny$\pm$0.07} & 0.31{\tiny$\pm$0.03}
    & 3.80{\tiny$\pm$1.48} & \textbf{1.53{\tiny$\pm$0.01}} & NaN\tnote{d} & 1.92{\tiny$\pm$0.04}
    & -\tnote{c} & - & - & - \\
& 3 & 39.67{\tiny$\pm$3.42} & \textbf{0.03{\tiny$\pm$0.002}} & 1.24{\tiny$\pm$0.08} & 0.34{\tiny$\pm$0.03}
    & OOM\tnote{d} & \textbf{1.53{\tiny$\pm$0.02}} & NaN & 1.88{\tiny$\pm$0.03}
    & - & - & - & - \\
& 4 & $>$ 1~\si{h} & \textbf{0.03{\tiny$\pm$0.006}} & 1.56{\tiny$\pm$0.05} & 0.32{\tiny$\pm$0.01}
    & OOM & \textbf{1.55{\tiny$\pm$0.01}} & NaN & 1.94{\tiny$\pm$0.05}
    & - & - & - & - \\
\midrule

\multirow{3}{*}{\makecell[c]{Third-Order\\System\\($n_x=3$)}}
& 2 & 33.17{\tiny$\pm$27.40} & \textbf{0.04{\tiny$\pm$0.004}} & 1.37{\tiny$\pm$0.12} & 0.34{\tiny$\pm$0.03}
    & $\approx$ 107\tnote{b} & 1.65{\tiny$\pm$0.06} & NaN & \textbf{1.29{\tiny$\pm$0.04}}
    & 0.04{\tiny$\pm$0.01} & \textbf{0.02{\tiny$\pm$0.0003}} & 0.43{\tiny$\pm$0.05} & 0.13{\tiny$\pm$0.002} \\
& 3 & $>$ 1~\si{h} & \textbf{0.06{\tiny$\pm$0.007}} & 1.82{\tiny$\pm$0.05} & 0.32{\tiny$\pm$0.03}
    & OOM & 4.51{\tiny$\pm$0.75} & NaN & \textbf{1.39{\tiny$\pm$0.03}}
    & 20.46{\tiny$\pm$1.20} & \textbf{0.02{\tiny$\pm$0.001}} & 0.46{\tiny$\pm$0.03} & 0.13{\tiny$\pm$0.007} \\
& 4 & $>$ 1~\si{h} & \textbf{0.08{\tiny$\pm$0.02}} & 2.41{\tiny$\pm$0.02} & 0.32{\tiny$\pm$0.01}
    & OOM & 10.78{\tiny$\pm$0.33} & NaN & \textbf{1.70{\tiny$\pm$0.05}}
    & $>$ 1~\si{h} & \textbf{0.03{\tiny$\pm$0.001}} & 0.48{\tiny$\pm$0.04} & 0.12{\tiny$\pm$0.005} \\
\midrule

\multirow{3}{*}{\makecell[c]{Cartpole\\($n_x=4$)}}
& 2 & $>$ 1~\si{h} & 0.57{\tiny$\pm$0.07} & 2.59{\tiny$\pm$0.03} & \textbf{0.35{\tiny$\pm$0.01}}
    & $>$ 4~\si{h} & 88.52{\tiny$\pm$9.94} & NaN & \textbf{13.28{\tiny$\pm$0.08}}
    & 0.17{\tiny$\pm$0.02} & \textbf{0.02{\tiny$\pm$0.001}} & 0.42{\tiny$\pm$0.03} & 0.13{\tiny$\pm$0.003} \\
& 3 & $>$ 1~\si{h} & 3.37{\tiny$\pm$0.30} & 4.66{\tiny$\pm$0.03} & \textbf{0.39{\tiny$\pm$0.01}}
    & $>$ 4~\si{h} & 474.10{\tiny$\pm$19.59} & NaN & \textbf{29.24{\tiny$\pm$1.77}}
    & 500.96{\tiny$\pm$81.82} & \textbf{0.03{\tiny$\pm$0.005}} & 0.47{\tiny$\pm$0.02} & 0.13{\tiny$\pm$0.004} \\
& 4 & $>$ 1~\si{h} & 6.70{\tiny$\pm$0.62} & 7.73{\tiny$\pm$0.05} & \textbf{0.51{\tiny$\pm$0.006}}
    & $>$ 4~\si{h} & 0.47{\tiny$\pm$0.07}~\si{h} & NaN & \textbf{78.25{\tiny$\pm$10.68}}
    & OOM & \textbf{0.03{\tiny$\pm$0.002}} & 0.50{\tiny$\pm$0.05} & 0.12{\tiny$\pm$0.004} \\
\midrule

\multirow{2}{*}{\makecell[c]{2D Quadrotor\\($n_x=6$)}}
& \multirow{2}{*}{2} 
& \multirow{2}{*}{$>$ 1~\si{h}} & \multirow{2}{*}{6.87{\tiny$\pm$0.19}} & \multirow{2}{*}{14.09{\tiny$\pm$1.89}} & \multirow{2}{*}{\textbf{0.79{\tiny$\pm$0.009}}}
& \multirow{2}{*}{$>$ 6~\si{h}} & \multirow{2}{*}{1.92{\tiny$\pm$0.33}~\si{h}} & \multirow{2}{*}{NaN} & \multirow{2}{*}{\textbf{506.32{\tiny$\pm$83.18}}}
& \multirow{2}{*}{$>$ 1~\si{h}} & \multirow{2}{*}{\textbf{0.24{\tiny$\pm$0.01}}} & \multirow{2}{*}{0.96{\tiny$\pm$0.03}} & \multirow{2}{*}{0.36{\tiny$\pm$0.04}} \\
& & & & & & & & & & & & & \\
\midrule

\multirow{2}{*}{\makecell[c]{SCARA Arm\\($n_x=8$)}}
& \multirow{2}{*}{2} 
& \multirow{2}{*}{ $>$ 1~\si{h}} & \multirow{2}{*}{153.22{\tiny$\pm$7.01}} & \multirow{2}{*}{11.21{\tiny$\pm$0.09}} &  \multirow{2}{*}{\textbf{9.05{\tiny$\pm$0.04}}}
& \multirow{2}{*}{$>$ 12~\si{h}} & \multirow{2}{*}{$>$ 12~\si{h}} & \multirow{2}{*}{ERR\tnote{d}} & \multirow{2}{*}{\textbf{9.53{\tiny$\pm$0.37}~\si{h}}}
& \multirow{2}{*}{$>$ 1~\si{h}} & \multirow{2}{*}{136.07{\tiny$\pm$48.93}} & \multirow{2}{*}{0.46{\tiny$\pm$0.04}} & \textbf{7.67{\tiny$\pm$2.58}} \\
& & & & & & & & & & & & & \textbf{0.34{\tiny$\pm$0.03}}\tnote{a} \\

\bottomrule\addlinespace[1ex]
\end{tabular}

\begin{tablenotes}\footnotesize
\item[a] As $\alpha,\beta$-CROWN only accepts linear constraints in problem definition, we relax all $\ell_2$ balls to $\ell_\infty$ balls in a way that favors $\alpha,\beta$-CROWN. $B_2(0,\eta)$ is replaced by $B_\infty(0,\eta)$ for checking the positive definiteness condition; $B_2(0,\mu)$ by $B_\infty(0,\mu)$ for decrease condition and by $B_\infty(0,\mu/\sqrt{n_x})$ for inclusion condition. Due to the large volume difference between $\ell_2$ and $\ell_\infty$ balls for $n_x=8$, we report two values for our method on GPU. The first uses the original $\ell_2$ constraints; the second uses the same $\ell_\infty$ relaxation used by $\alpha,\beta$-CROWN to provide a direct comparison.
\item[b] dReal times out ($>$ 1~\si{h}) on two out of four trials for the inverted pendulum and three out of four for the third-order system.
\item[c] The inclusion condition is omitted for the inverted pendulum and unicycle path following as $\widehat{V}_r \geq \gamma_V \mu$.
\item[d] NaN: encounters NaN in bound evaluation. ERR: internal error. OOM: out of memory.
\end{tablenotes}
\end{threeparttable}
\end{table*}

\begin{table*}[t]
\centering
\caption{Verification time (in seconds unless otherwise specified) of the positive definiteness, decrease, and inclusion (of $B_2(0, \mu)$ by $\Omega_{\widehat{V}_r}$) conditions for different systems and different Lipschitz bounds of $V_{\theta_1}$ (denoted by $\gamma_V$).}
\label{tab:time_vs_lip}
\begin{threeparttable}
\begin{tabular}{C{1.5cm} C{0.8cm} C{1.39cm} C{1.39cm} C{1.39cm} C{1.39cm} C{1.39cm} C{1.39cm} C{1.39cm}}
\toprule
\multirow{2}{*}{System} & \multirow{2}{*}{\makecell{$\gamma_V$}} & \multirow{2}{*}{Success Rate} & \multicolumn{2}{c}{Positive Definiteness} & \multicolumn{2}{c}{Decrease} & \multicolumn{2}{c}{Inclusion} \\
\cmidrule(lr){4-5}\cmidrule(lr){6-7}\cmidrule(lr){8-9}
& & &  Ours (CPU) & Ours (GPU) & Ours (CPU) & Ours (GPU) & Ours (CPU) & Ours (GPU) \\
\midrule
\multirow{3}{*}{\makecell{Inverted\\Pendulum}}  
& 1.0 & 100\% & 0.02{\tiny$\pm$0.003} & 0.31{\tiny$\pm$0.02} & 0.62{\tiny$\pm$0.23} & 0.72{\tiny$\pm$0.03} & - & - \\
& 2.0 & 100\% & 0.02{\tiny$\pm$0.002} & 0.32{\tiny$\pm$0.02} & 0.80{\tiny$\pm$0.006} & 0.73{\tiny$\pm$0.04} & - & - \\
& 4.0 & 100\% & 0.02{\tiny$\pm$0.003} & 0.31{\tiny$\pm$0.03} & 0.80{\tiny$\pm$0.005} & 0.71{\tiny$\pm$0.02} & - & - \\
\midrule

\multirow{3}{*}{\makecell{Unicycle Path\\Following}}  
& 1.0 & 100\% & 0.02{\tiny$\pm$0.0004} & 0.31{\tiny$\pm$0.03} & 1.53{\tiny$\pm$0.01} & 1.92{\tiny$\pm$0.04} & - & - \\
& 2.0 & 100\% & 0.02{\tiny$\pm$0.0008} & 0.29{\tiny$\pm$0.005} & 1.54{\tiny$\pm$0.01} & 1.91{\tiny$\pm$0.03} & - & - \\
& 4.0 & 100\% & 0.02{\tiny$\pm$0.001} & 0.31{\tiny$\pm$0.02} & 1.53{\tiny$\pm$0.02} & 1.92{\tiny$\pm$0.07} & - & - \\
\midrule

\multirow{3}{*}{\makecell{Third-Order\\System}}  
& 2.5 & 100\% & 0.04{\tiny$\pm$0.003} & 0.32{\tiny$\pm$0.03} & 1.49{\tiny$\pm$0.01} & 1.26{\tiny$\pm$0.02} & 0.02{\tiny$\pm$0.0003} & 0.13{\tiny$\pm$0.003} \\
& 5.0 & 100\% & 0.04{\tiny$\pm$0.004} & 0.34{\tiny$\pm$0.03} & 1.65{\tiny$\pm$0.06} & 1.29{\tiny$\pm$0.04} & 0.02{\tiny$\pm$0.0003} & 0.13{\tiny$\pm$0.002} \\
& 10.0 & 100\% & 0.05{\tiny$\pm$0.004} & 0.32{\tiny$\pm$0.01} & 2.06{\tiny$\pm$0.55} & 1.21{\tiny$\pm$0.24} & 0.02{\tiny$\pm$0.0002} & 0.13{\tiny$\pm$0.007} \\
\midrule

\multirow{3}{*}{\makecell{Cartpole}}  
& 1.0 & 100\% & 0.60{\tiny$\pm$0.05} & 0.35{\tiny$\pm$0.02} & 82.61{\tiny$\pm$8.62} & 11.09{\tiny$\pm$2.81} & 0.02{\tiny$\pm$0.003} & 0.13{\tiny$\pm$0.007} \\
& 2.0 & 100\% & 0.62{\tiny$\pm$0.04} & 0.38{\tiny$\pm$0.04} & 108.69{\tiny$\pm$12.19} & 14.73{\tiny$\pm$0.44} & 0.02{\tiny$\pm$0.004} & 0.13{\tiny$\pm$0.003} \\
& 4.0 & 100\% & 0.65{\tiny$\pm$0.06} & 0.36{\tiny$\pm$0.02} & 146.89{\tiny$\pm$20.81} & 17.60{\tiny$\pm$1.23} & 0.02{\tiny$\pm$0.001} & 0.13{\tiny$\pm$0.003} \\
\midrule

\multirow{2}{*}{\makecell{2D Quadrotor}}  
& 0.5 & 100\% & 5.24{\tiny$\pm$0.08} & 0.68{\tiny$\pm$0.04} & 1.06{\tiny$\pm$0.03}~\si{h} & 292.50{\tiny$\pm$8.23} & 0.33{\tiny$\pm$0.07} & 0.34{\tiny$\pm$0.004} \\
& 1.0 & 100\% & 6.87{\tiny$\pm$0.19} & 0.79{\tiny$\pm$0.009} & 1.92{\tiny$\pm$0.33}~\si{h} & 506.32{\tiny$\pm$83.18} & 0.24{\tiny$\pm$0.01} & 0.36{\tiny$\pm$0.04} \\
\midrule

\multirow{2}{*}{\makecell{SCARA Arm}}  
& 0.5 & 100\% & 115.60{\tiny$\pm$2.46} & 7.57{\tiny$\pm$0.07} & $>$ 12~\si{h} & 5.12{\tiny$\pm$0.33}~\si{h} & 17.94{\tiny$\pm$0.67} & 1.26{\tiny$\pm$0.05} \\
& 1.0 & 100\% & 153.22{\tiny$\pm$7.01} & 9.05{\tiny$\pm$0.04} & $>$ 12~\si{h} & 9.53{\tiny$\pm$0.37}~\si{h} & 136.07{\tiny$\pm$48.93} & 7.67{\tiny$\pm$2.58} \\
\bottomrule\addlinespace[1ex]
\end{tabular}
\end{threeparttable}
\end{table*}

\section{Simulation Studies}\label{sec:simulation_studies}
We evaluate our methodology on six dynamical systems (five depicted in Fig.~\ref{fig:system_diagrams}): the inverted pendulum, unicycle path following, a third-order system, the cartpole, the 2D quadrotor, and the SCARA arm. Table~\ref{tab:time} reports the RCLF verification times for Algorithm~\ref{alg:bnb}, dReal \cite{gao2013dreal} (a CPU-based SMT solver widely used for Lyapunov verification \cite{chang2019neural, zhou2022neural, liu2025physics}), and $\alpha,\beta$-CROWN \cite{wang2021beta} (a GPU-based bound-propagation verifier, also used in \cite{yang2024lyapunov}) at the same precision. For each case we train with four random seeds and report the mean and standard deviation. For the SCARA arm, we use interval bound propagation (IBP) for local bounds on the gradient, Hessian, and third-order derivative of the LNNs to accelerate verification. Our results show that Algorithm~\ref{alg:bnb} offers better scalability for larger LNNs and significantly reduces the computational time. 

Table~\ref{tab:time_vs_lip} fixes $L_V=2$ and varies $\gamma_V$, reporting the verification time and success rate (over four random seeds) for each example. We use Sandwich Layers \cite{wang2023direct} for this section and set $\Psi^{(\ell)} = I$ for $V_{\theta_1}$ throughout. The network structures and hyperparameters are given in Tables~\ref{tab:nn_structure} and \ref{tab:hyperparams}. {\color{revcolor} Appendix~\ref{appx:additional_results_on_other_layers} reports results with the Spectral Normalization \cite{miyato2018spectral} and Orthogonal Layer \cite{trockman2021orthogonalizing} architectures, showing that the speedups are consistent and independent of the underlying Lipschitz parameterization.} All verifications are run on a PC with 64~GB RAM, an Intel Core i9-13900K, and an NVIDIA GeForce RTX 4090 GPU.

\subsection{Inverted Pendulum}
\begin{figure}[t]
    \centering
    \begin{subfigure}[b]{0.2\textwidth}
        \centering
        \includegraphics[width=\textwidth]{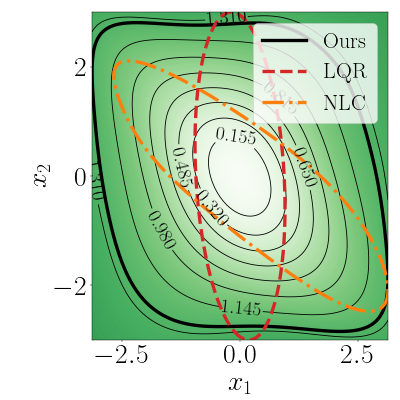}
        \caption{$V_{\theta_1}(x)$}
        \label{fig:eg_inv_pend_forward_inv}
    \end{subfigure}
    \begin{subfigure}[b]{0.2\textwidth}
        \centering
        \includegraphics[width=\textwidth]{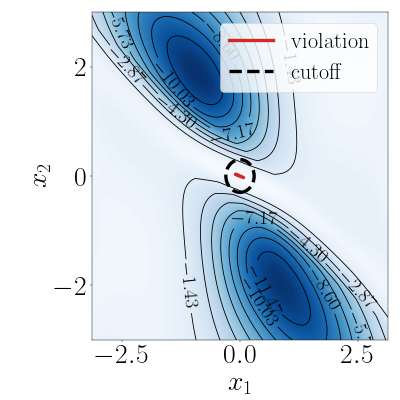}
        \caption{$H(x)+\omega(x)$}
        \label{fig:eg_inv_pend_stability}
    \end{subfigure}
    \caption{Visualization of $V_{\theta_1}(x)$ and $H(x)+\omega(x)$ for the inverted pendulum. \subref{fig:eg_inv_pend_forward_inv}: level sets of $V_{\theta_1}(x)$ and the largest forward invariant sets given by our method, LQR, and NLC. \subref{fig:eg_inv_pend_stability}: level sets of $H(x)+\omega(x)$ with the cutoff region $B_2(0,\mu)$, showing that violations ($H(x)+\omega(x)>0$) remain confined within the cutoff region, consistent with condition~\eqref{eq:stability_cond}.}
    \label{fig:eg_inv_pend_contour}
\end{figure}

\begin{table}[t]
\centering
\caption{Forward invariant region ratio (\%) for inverted pendulum.}
\label{tab:eg_inv_pend_ratio}
\begin{tabular}{ccccc}
\toprule
LQR & NLC & Ours ($L_V{=}2$) & Ours ($L_V{=}3$) & Ours ($L_V{=}4$) \\
\midrule
23.28 & 28.57 & 71.86 & 74.24 & 75.15 \\
\bottomrule
\end{tabular}
\end{table}

We consider the stationary inverted pendulum system (see Fig.~\ref{fig:inverted_pendulum_diagram}) with dynamics
\begin{equation}
    \dot{x}_1 = x_2, \quad \dot{x}_2 = [m g l \sin(x_1) + u + d(t,x)]/m l^2
\end{equation}
where $x_1$ represents the angle between the pendulum and the upward direction and $x_2$ the angular velocity with the state space being $\mathcal{X} = \{x \in \R^2 \mid |x_1| \leq \pi, |x_2| \leq 3\}$. The input $u \in \R$ is the torque applied to the pendulum with $\mathcal{U} = [-15, 15]$. We consider a disturbance signal $d(t,x) = d_0 \sin(2 \pi t) - d_1 x_2$ where $d_0 \sin(2 \pi t)$ represents a time-varying external torque and $- d_1 x_2$ the viscous friction proportional to the angular velocity. The system parameters are $m=1.0$~\si{kg}, $l = 1.0$~\si{m}, $g = 9.81$~\si{m.s^{-2}}, $d_0 = 0.1$~\si{N.m}, and $d_1 = 0.1$~\si{N.m.s}. As shown in Fig.~\ref{fig:eg_inv_pend_forward_inv} and Table~\ref{tab:eg_inv_pend_ratio}, our method achieves a much larger forward invariant region than LQR and NLC \cite{chang2019neural}. The region where $H(x)+\omega(x) > 0$ is contained in $B_2(0, \mu)$ with $\mu = 0.3$ as depicted in Fig.~\ref{fig:eg_inv_pend_stability}. The verified forward invariant region is $\Omega_{\widehat{V}_r}$ with $\widehat{V}_r = 1.19$ which is greater than $\gamma_V \mu$ in this case, so $B_2(0, \mu)$ must be contained within $\Omega_{\widehat{V}_r}$. We sample 200 initial conditions on the boundary of $\Omega_{\widehat{V}_r}$, and the average control effort ($\int_0^T \lVert u(t) \rVert_2^2 dt$) of LQR, NLC, and our method is 353.93, 302.63, and 255.95, respectively. The average settling time ($\inf \{ \tau \geq 0 \mid \lVert x(t) \rVert_2 \leq 0.05, \forall t \geq \tau \}$) is 11.38~\si{s}, 8.45~\si{s}, and 9.65~\si{s}, respectively.

\begin{figure}[t]
    \centering
    \begin{subfigure}[b]{0.2\textwidth}
        \centering
        \includegraphics[width=\textwidth]{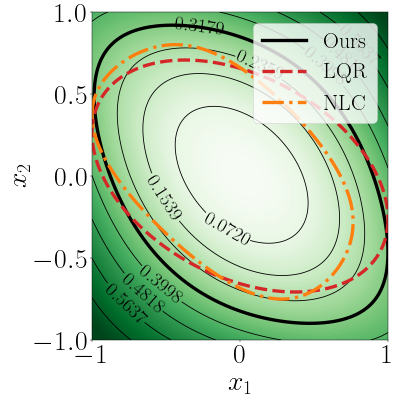}
        \caption{$V_{\theta_1}(x)$}
        \label{fig:eg_unicycle_forward_inv}
    \end{subfigure}
    \begin{subfigure}[b]{0.2\textwidth}
        \centering
        \includegraphics[width=\textwidth]{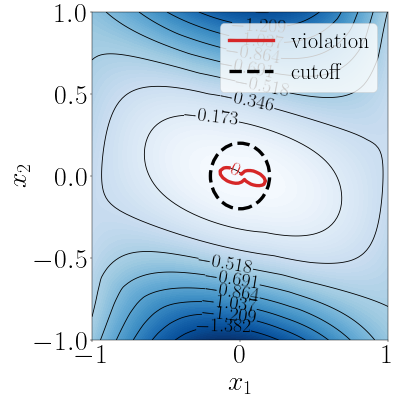}
        \caption{$H(x) + \omega(x)$}
        \label{fig:eg_unicycle_stability}
    \end{subfigure}
    \caption{Visualization of $V_{\theta_1}(x)$ and $H(x)+\omega(x)$ for the unicycle path following example. \subref{fig:eg_unicycle_forward_inv}: level sets of $V_{\theta_1}(x)$ and the largest forward invariant sets given by our method, LQR, and NLC. \subref{fig:eg_unicycle_stability}: level sets of $H(x)+\omega(x)$ and the cutoff region $B_2(0, \mu)$.}
    \label{fig:eg_unicycle_contour}
\end{figure}

\begin{table}[t]
\centering
\caption{Forward invariant region ratio (\%) for unicycle.}
\label{tab:eg_unicycle_ratio}
\begin{tabular}{ccccc}
\toprule
LQR & NLC & Ours ($L_V{=}2$) & Ours ($L_V{=}3$) & Ours ($L_V{=}4$) \\
\midrule
50.58 & 44.09 & 62.06 & 62.51 & 63.88 \\
\bottomrule
\end{tabular}
\end{table}

\subsection{Unicycle Path Following}
Consider the unicycle path following example depicted in Fig.~\ref{fig:unicycle_diagram}. Denote by $\theta_p$ the angle between the current tangent to the path and the horizontal direction and $\theta_u$ the angle between the unicycle velocity and the horizontal direction. Let $x_1$ be the distance between the center of the unicycle and the path and $x_2 = \theta_u - \theta_p$ the orientation error. In the case of following a circular path \cite{de2005feedback}, the dynamics are given by
\begin{equation}\label{eq:unicycle_dynamics}
    \dot{x}_1 = v \sin (x_2), \quad 
    \dot{x}_2 = u - \frac{v \cos (x_2)}{R - x_1} + d(t)
\end{equation}
where $v \in \R_+$ is the fixed linear velocity of the unicycle, $R \in \R_+$ is the radius of the circular path, and the control $u \in \R$ is the angular velocity of the unicycle. The disturbance signal $d(t)$ represents the uncertainties in the dynamics of orientation error (e.g., input noise). Assume that $d(t)$ takes the form $d(t) = d_0 \cos(2 \pi t)$, and the system parameters are $v = 2$~\si{m.s^{-1}}, $R = 10$~\si{m}, and $d_0 = 0.1$~\si{s^{-1}}. Let the state space be $\mathcal{X} = \{x \in \R^2 \mid |x_1| \leq 1, |x_2| \leq 1\}$ and control limits $\mathcal{U} = [-3.8, 4.2]$. As we still need to generate a control equal to $v/R$ at the equilibrium point $(0,0)$, we apply a change of variable $u = v/R + \tilde{u}$ to shift the control to zero at $x=0$.

We see from Fig.~\ref{fig:eg_unicycle_forward_inv} and Table~\ref{tab:eg_unicycle_ratio} that our method obtains a larger forward invariant region than LQR and NLC. As shown in Fig.~\ref{fig:eg_unicycle_stability}, the region where $H(x)+\omega(x) > 0$ is contained in $B_2(0, \mu)$ with $\mu = 0.2$. As the verified forward invariant region is $\Omega_{\widehat{V}_r}$ with $\widehat{V}_r = 0.32$ which is greater than $\gamma_V \mu$ in this case, $B_2(0, \mu)$ must be contained within $\Omega_{\widehat{V}_r}$. 
We sample 200 initial conditions on the boundary of $\Omega_{\widehat{V}_r}$, and the average control effort ($\int_0^T \lVert \tilde{u}(t) \rVert_2^2 dt$) of LQR, NLC, and our method is 0.84, 7.12, and 0.80, respectively. The average settling time ($\inf \{ \tau \geq 0 \mid \lVert x(t) \rVert_2 \leq 0.05, \forall t \geq \tau \}$) is 4.16~\si{s}, 4.75~\si{s}, and 2.52~\si{s}, respectively.

\begin{figure}[t]
    \centering
    \begin{subfigure}[b]{0.20\textwidth}
        \centering
        \includegraphics[width=\textwidth]{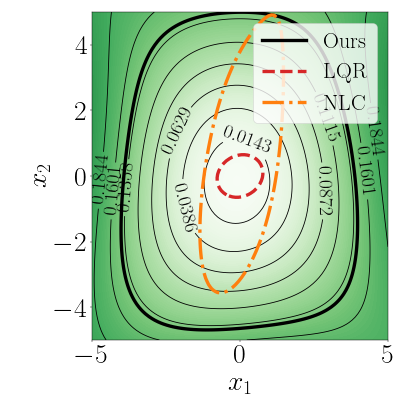}
        \caption{$V_{\theta_1}(x)$ on $(x_1, x_2)$}
        \label{fig:eg_tac_third_order_sys_forward_inv_1}
    \end{subfigure}
    \begin{subfigure}[b]{0.20\textwidth}
        \centering
        \includegraphics[width=\textwidth]{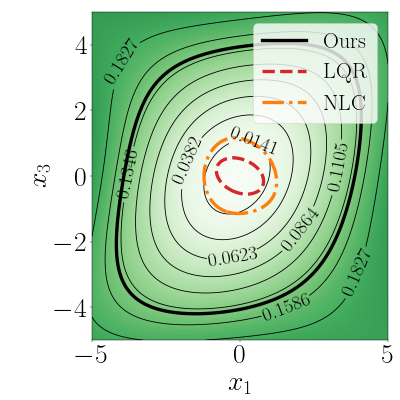}
        \caption{$V_{\theta_1}(x)$ on $(x_1, x_3)$}
        \label{fig:eg_tac_third_order_sys_forward_inv_2}
    \end{subfigure}
    \begin{subfigure}[b]{0.20\textwidth}
        \centering
        \includegraphics[width=\textwidth]{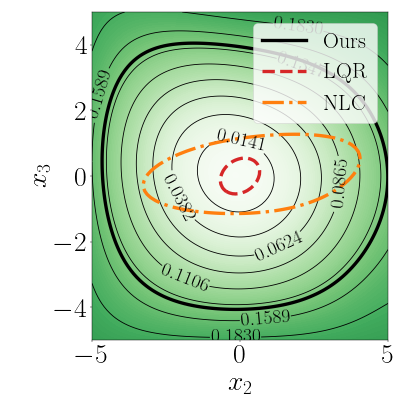}
        \caption{$V_{\theta_1}(x)$ on $(x_2, x_3)$}
        \label{fig:eg_tac_third_order_sys_forward_inv_3}
    \end{subfigure}
    \begin{subfigure}[b]{0.20\textwidth}
        \centering
        \includegraphics[width=\textwidth]{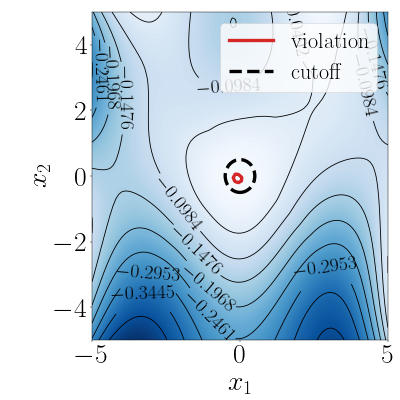}
        \caption{$H(x)+\omega(x)$ on $(x_1, x_2)$}
        \label{fig:eg_tac_third_order_sys_stability_1}
    \end{subfigure}
    \begin{subfigure}[b]{0.20\textwidth}
        \centering
        \includegraphics[width=\textwidth]{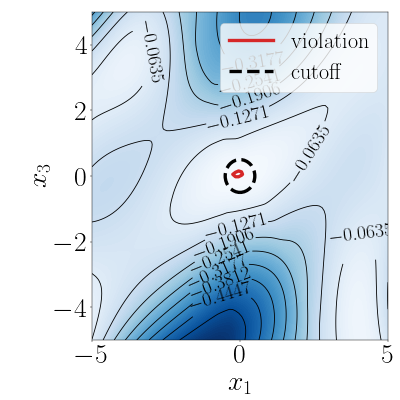}
        \caption{$H(x)+\omega(x)$ on $(x_1, x_3)$}
        \label{fig:eg_tac_third_order_sys_stability_2}
    \end{subfigure}
    \begin{subfigure}[b]{0.20\textwidth}
        \centering
        \includegraphics[width=\textwidth]{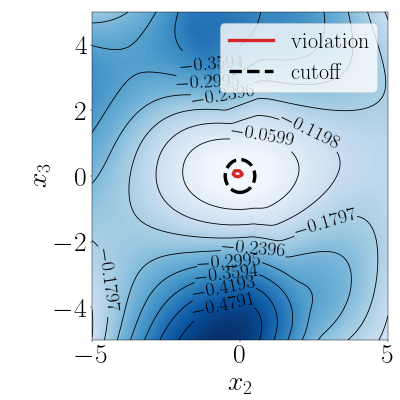}
        \caption{$H(x)+\omega(x)$ on $(x_2, x_3)$}
        \label{fig:eg_tac_third_order_sys_stability_3}
    \end{subfigure}
    \caption{Visualization of $V_{\theta_1}(x)$ and $H(x)+\omega(x)$ for the third-order system. \subref{fig:eg_tac_third_order_sys_forward_inv_1}-\subref{fig:eg_tac_third_order_sys_forward_inv_3}: level sets of $V_{\theta_1}(x)$ and the largest forward invariant sets given by our method, LQR, and NLC. \subref{fig:eg_tac_third_order_sys_stability_1}-\subref{fig:eg_tac_third_order_sys_stability_3}: level sets of $H(x)+\omega(x)$ and the cutoff region $B_2(0, \mu)$.}
    \label{fig:eg_tac_third_order_sys_contour}
\end{figure}

\begin{table}[t]
\centering
\caption{Forward invariant region ratio (\%) for the third-order system.}
\label{tab:eg_tac_third_order_sys_ratio}
\begin{tabular}{ccccc}
\toprule
LQR & NLC & Ours ($L_V{=}2$) & Ours ($L_V{=}3$) & Ours ($L_V{=}4$) \\
\midrule
0.12 & 2.59 & 39.78 & 40.19 & 41.89 \\
\bottomrule
\end{tabular}
\end{table}

\subsection{Third-Order System}\label{sec:third_order_sys}
Consider the third-order control-affine system proposed in \cite{andrieu2010uniting}. The dynamics of this system are given by
\begin{subequations}
    \begin{align}
    \dot{x}_1 &= -x_1 + x_3, \\
    \dot{x}_2 &= x_1^2 - x_2 -2 x_1 x_3 + x_3, \\
    \dot{x}_3 &= -x_2 + u + d(t)
    \end{align}
\end{subequations}
where the disturbance signal satisfies $\lvert d(t) \rvert \leq 0.1$. Let the state space be $\mathcal{X} = \{x \in \R^3 \mid |x_1| \leq 5, |x_2| \leq 5, |x_3| \leq 5\}$ and control limits $\mathcal{U} = \{u \in \R \mid |u| \leq 20\}$. 

We visualize the forward invariant regions given by our method, LQR, and NLC in Figs.~\ref{fig:eg_tac_third_order_sys_forward_inv_1}-\ref{fig:eg_tac_third_order_sys_forward_inv_3}. As shown in Figs.~\ref{fig:eg_tac_third_order_sys_stability_1}-\ref{fig:eg_tac_third_order_sys_stability_3}, the region where $H(x)+\omega(x) > 0$ is contained in $B_2(0, \mu)$ with $\mu = 0.5$. Table~\ref{tab:eg_tac_third_order_sys_ratio} shows that our method achieves a larger forward invariant region than both LQR and NLC, and we note that the LQR solution only gives a locally stable controller with a small forward invariant region. We sample 200 initial conditions on the boundary of $\Omega_{\widehat{V}_r}$, and under $d(t) = 0.1 \sin(2 \pi t)$, the average control effort ($\int_0^T \lVert \tilde{u}(t) \rVert_2^2 dt$) of NLC and our method is 25.81 and 12.59, respectively. The average settling time ($\inf \{ \tau \geq 0 \mid \lVert x(t) \rVert_2 \leq 0.05, \forall t \geq \tau \}$) is 4.81~\si{s} and 4.61~\si{s}, respectively. Five out of the 200 trajectories are unstable under the LQR controller.

\begin{figure}[t]
    \centering
    \begin{subfigure}[b]{0.20\textwidth}
        \centering
        \includegraphics[width=\textwidth]{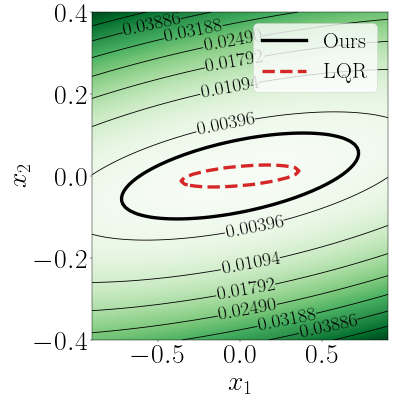}
        \caption{$V_{\theta_1}(x)$ on $(x_1, x_2)$}
        \label{fig:eg_cartpole_forward_inv_1}
    \end{subfigure}
    \begin{subfigure}[b]{0.20\textwidth}
        \centering
        \includegraphics[width=\textwidth]{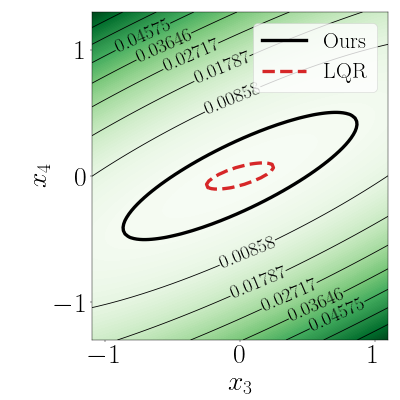}
        \caption{$V_{\theta_1}(x)$ on $(x_3, x_4)$}
        \label{fig:eg_cartpole_forward_inv_2}
    \end{subfigure}
    \begin{subfigure}[b]{0.20\textwidth}
        \centering
        \includegraphics[width=\textwidth]{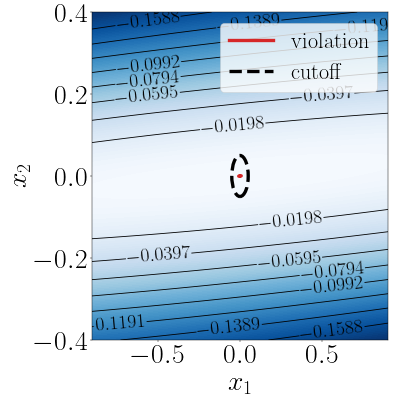}
        \caption{$H(x)+\omega(x)$ on $(x_1, x_2)$}
        \label{fig:eg_cartpole_stability_1}
    \end{subfigure}
    \begin{subfigure}[b]{0.20\textwidth}
        \centering
        \includegraphics[width=\textwidth]{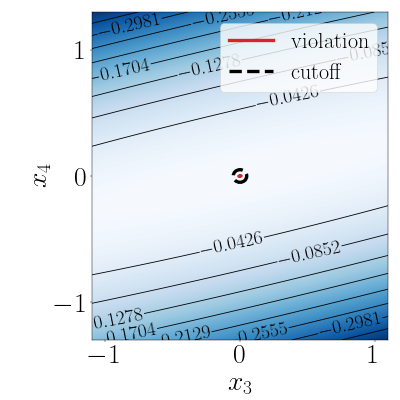}
        \caption{$H(x)+\omega(x)$ on $(x_3, x_4)$}
        \label{fig:eg_cartpole_stability_2}
    \end{subfigure}
    \caption{Visualization of $V_{\theta_1}(x)$ and $H(x)+\omega(x)$ for the cartpole. \subref{fig:eg_cartpole_forward_inv_1} and \subref{fig:eg_cartpole_forward_inv_2}: level sets of $V_{\theta_1}(x)$ and the largest forward invariant sets given by our method and LQR. \subref{fig:eg_cartpole_stability_1} and \subref{fig:eg_cartpole_stability_2}: level sets of $H(x)+\omega(x)$ and the cutoff region $B_2(0, \mu)$.}
    \label{fig:eg_cartpole_contour}
\end{figure}

\begin{table}[t]
\centering
\caption{Forward invariant region ratio (\%) for cartpole.}
\label{tab:eg_cartpole_ratio}
\begin{tabular}{cccc}
\toprule
LQR & Ours ($L_V{=}2$) & Ours ($L_V{=}3$) & Ours ($L_V{=}4$) \\
\midrule
{0.19} & 5.36 & 5.40 & 5.43 \\
\bottomrule
\end{tabular}
\end{table}

\subsection{Cartpole}
Consider the cartpole system in Fig.~\ref{fig:cartpole_diagram} with dynamics 
{\allowdisplaybreaks
\begin{subequations}
\begin{align}
\dot{x}_1 &= x_3, \quad \dot{x}_2 = x_4, \\
\dot{x}_3 &= \frac{m_p \sin(x_2) (-l x_4^2 + g \cos(x_2)) + u + d(x)}{m_c + m_p \sin^2(x_2)}, \\
\dot{x}_4 &= \frac{1}{l(m_c + m_p \sin^2(x_2))}[(m_p + m_c) g \sin(x_2)  \notag \\
&\quad - m_p l \sin(x_2) \cos(x_2) x_4^2 + \cos(x_2) (u + d(x))]
\end{align}
\end{subequations}
}where $x_1$ is the horizontal position of the cart, $x_2$ the angle between the pole and the upward direction in the counterclockwise direction, $x_3$ the horizontal speed of the cart, $x_4$ the pole’s angular velocity, and $u$ the force applied on the cart in the horizontal direction. $d(x)$ models the friction between the cart and the ground, and $d(x) = -d_0 \sign(x_3) (m_c+m_p) g$. Let the state space be $\mathcal{X} = \{x \in \R^4 \mid |x_1| \leq 0.9, |x_2| \leq 0.4, |x_3| \leq 1.1, |x_4| \leq 1.4\}$ and control limits $\mathcal{U} = [-80,80]$. The system parameters are $m_c = 2$~\si{kg}, $m_p = 0.1$~\si{kg}, $l=1$~\si{m}, $g = 9.81$~\si{m.s^{-2}}, and the friction coefficient is $d_0 = 0.001$. In this case, $\bigl \lvert \frac{d(x)}{m_c + m_p \sin^2(x_2)} \bigr \rvert \leq \bigl \lvert \frac{d(x)}{m_c} \bigr \rvert = 0.0104$ and $\bigl \lvert \frac{\cos(x_2) d(x)}{l(m_c + m_p \sin^2(x_2))} \bigr \rvert \leq \bigl \lvert \frac{d(x)}{l m_c} \bigr \rvert = 0.0104$.

Figures~\ref{fig:eg_cartpole_forward_inv_1} and \ref{fig:eg_cartpole_forward_inv_2} illustrate the forward invariant regions of our method and LQR. Figures~\ref{fig:eg_cartpole_stability_1} and \ref{fig:eg_cartpole_stability_2} demonstrate that the set where $H(x)+\omega(x)>0$ is entirely contained within the ball $B_2(0, \mu)$ with $\mu = 0.05$. Table~\ref{tab:eg_cartpole_ratio} shows that our method achieves a larger forward invariant region than LQR. The training of NLC did not converge within a reasonable time for systems with four or more dimensions. We sample 200 initial conditions on the boundary of $\Omega_{\widehat{V}_r}$. With these same initial conditions, the average control effort ($\int_0^T \lVert u(t) \rVert_2^2 dt$) of LQR and our method is 2.43 and 3.04, respectively. The average settling time ($\inf \{ \tau \geq 0 \mid \lVert x(t) \rVert_2 \leq 0.05, \forall t \geq \tau \}$) is 5.56~\si{s} and 3.34~\si{s}, respectively.

\begin{figure}[t]
    \centering
    \begin{subfigure}[b]{0.20\textwidth}
        \centering
        \includegraphics[width=\textwidth]{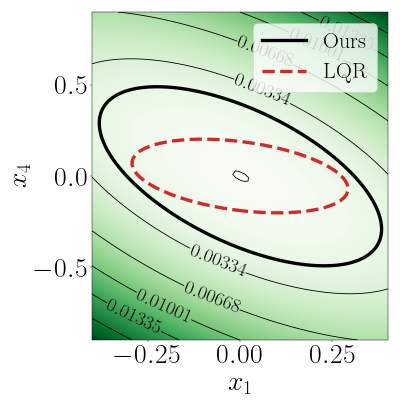}
        \caption{$V_{\theta_1}(x)$ on $(x_1, x_4)$}
        \label{fig:eg_quad_2d_forward_inv_1}
    \end{subfigure}
    \begin{subfigure}[b]{0.20\textwidth}
        \centering
        \includegraphics[width=\textwidth]{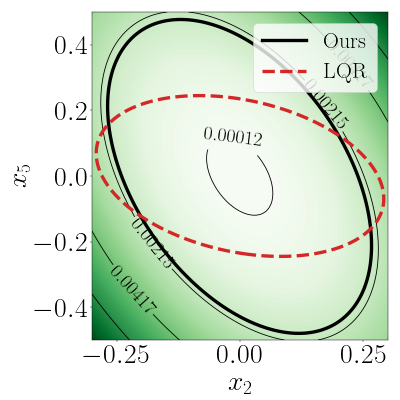}
        \caption{$V_{\theta_1}(x)$ on $(x_2, x_5)$}
        \label{fig:eg_quad_2d_forward_inv_2}
    \end{subfigure}
    \begin{subfigure}[b]{0.20\textwidth}
        \centering
        \includegraphics[width=\textwidth]{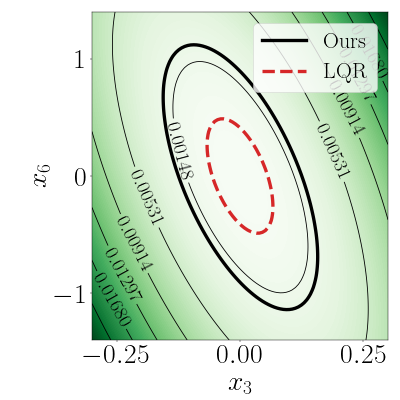}
        \caption{$V_{\theta_1}(x)$ on $(x_3, x_6)$}
        \label{fig:eg_quad_2d_forward_inv_3}
    \end{subfigure} 
    \begin{subfigure}[b]{0.20\textwidth}
        \centering
        \includegraphics[width=\textwidth]{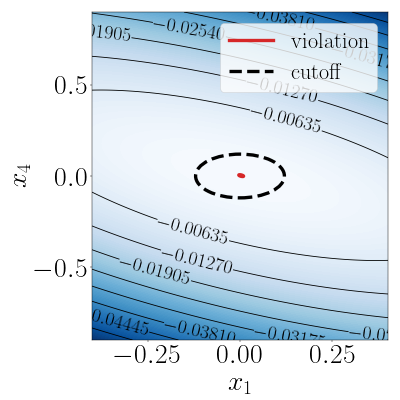}
        \caption{$H(x)+\omega(x)$ on $(x_1, x_4)$}
        \label{fig:eg_quad_2d_stability_1}
    \end{subfigure}
    \begin{subfigure}[b]{0.20\textwidth}
        \centering
        \includegraphics[width=\textwidth]{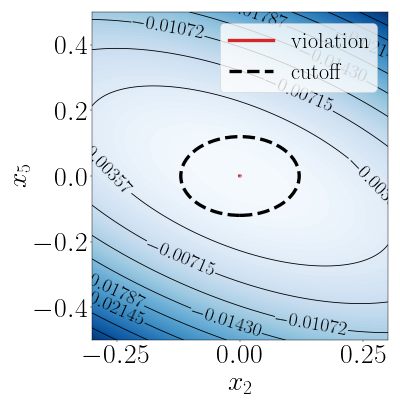}
        \caption{$H(x)+\omega(x)$ on $(x_2, x_5)$}
        \label{fig:eg_quad_2d_stability_2}
    \end{subfigure}
    \begin{subfigure}[b]{0.20\textwidth}
        \centering
        \includegraphics[width=\textwidth]{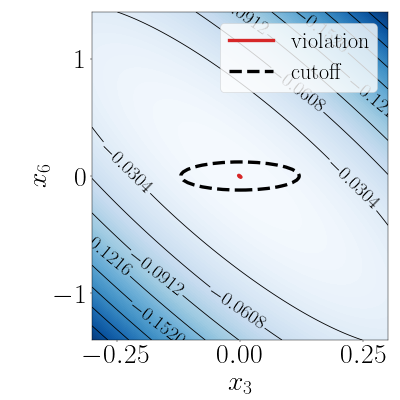}
        \caption{$H(x)+\omega(x)$ on $(x_3, x_6)$}
        \label{fig:eg_quad_2d_stability_3}
    \end{subfigure}
    \caption{Visualization of $V_{\theta_1}(x)$ and $H(x)+\omega(x)$ for the 2D quadrotor. \subref{fig:eg_quad_2d_forward_inv_1}-\subref{fig:eg_quad_2d_forward_inv_3}: level sets of $V_{\theta_1}(x)$ and the largest forward invariant sets given by our method and LQR. \subref{fig:eg_quad_2d_stability_1}-\subref{fig:eg_quad_2d_stability_3}: level sets of $H(x)+\omega(x)$ and the cutoff region $B_2(0, \mu)$.}
    \label{fig:eg_quad_2d_contour}
\end{figure}

\subsection{2D Quadrotor}
Consider the 2D quadrotor given in Fig.~\ref{fig:quadrotor_2d_diagram} with dynamics
\begin{subequations}
\begin{align}
    \dot{x}_1 &= x_4, \quad \dot{x}_2 = x_5, \quad \dot{x}_3 = x_6, \allowbreak \\
    \dot{x}_4 &= -\sin(x_3) (u_1+u_2)/m + d(t), \\
    \dot{x}_5 &= -g + \cos(x_3) (u_1+u_2)/m, \\
    \dot{x}_6 &= l (-u_1 + u_2)/I
\end{align}
\end{subequations}
where $x_1$ and $x_2$ are the horizontal and vertical positions of the quadrotor, $x_3$ is the angle between the body of the quadrotor and the horizontal direction, $x_4$ and $x_5$ are the linear velocities in the horizontal and vertical directions, $x_6$ is the angular velocity, and $u_1$ and $u_2$ are the forces provided by the propellers. Let $m$ be the mass, $I$ the inertia, and $l$ the arm length of the quadrotor. $d(t)$ models the wind effect in the horizontal direction, and we assume that $|d(t)| \leq 0.003$. Let the state space be $\mathcal{X} = \{x \in \R^6 \mid |x_1| \leq 0.4, |x_2| \leq 0.3, |x_3| \leq 0.3, |x_4| \leq 0.9, |x_5| \leq 0.5, |x_6| \leq 1.4\}$ and control limits $\mathcal{U} = \{u \in \R^2 \mid -7.19 \leq u_1, u_2 \leq 26.81\}$. The system parameters are $m = 2$~\si{kg}, $I = 0.2$~\si{kg.m^2}, $l=0.4$~\si{m}, and $g = 9.81$~\si{m.s^{-2}}. We apply a change of variables $u_1 = m g / 2+ \tilde{u}_1$ and $u_2 = m g / 2+ \tilde{u}_2$ to shift the controls to zero at $x=0$.

We visualize the forward invariant regions given by our method and LQR in Figs.~\ref{fig:eg_quad_2d_forward_inv_1}-\ref{fig:eg_quad_2d_forward_inv_3}. Figures~\ref{fig:eg_quad_2d_stability_1}-\ref{fig:eg_quad_2d_stability_3} show that the region where $H(x)+\omega(x) > 0$ is contained in $B_2(0, \mu)$ with $\mu = 0.12$.
We sample 200 initial conditions on the boundary of $\Omega_{\widehat{V}_r}$, and the average control effort ($\int_0^T \lVert \tilde{u}(t) \rVert_2^2 dt$) of LQR and our method is 0.38 and 0.43, respectively. The average settling time ($\inf \{ \tau \geq 0 \mid \lVert x(t) \rVert_2 \leq 0.05, \forall t \geq \tau \}$) is 4.81~\si{s} and 1.95~\si{s}, respectively.

\begin{figure}[t]
    \centering
    \begin{subfigure}[b]{0.18\textwidth}
        \centering
        \includegraphics[width=\textwidth]{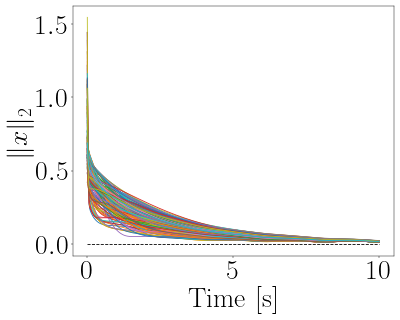}
        \caption{$\ell_2$ norm of states $x$}
        \label{fig:eg_scara_state_norm_perturbed}
    \end{subfigure}
    \begin{subfigure}[b]{0.18\textwidth}
        \centering
        \includegraphics[width=\textwidth]{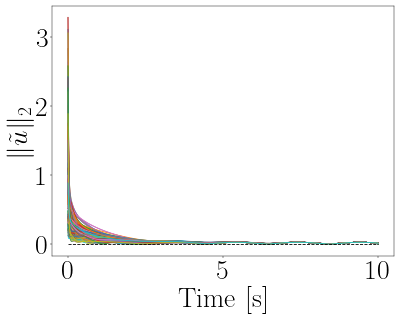}
        \caption{$\ell_2$ norm of $\tilde{u}$}
        \label{fig:eg_scara_control_perturbed}
    \end{subfigure}
    \caption{State and control trajectories for the SCARA arm.}
    \label{fig:eg_scara_sim}
\end{figure}

\subsection{SCARA Arm}
Consider the 8-state SCARA arm in Fig.~\ref{fig:scara_diagram} with dynamics given in \cite[Example 4.4]{murray2017mathematical}. The SCARA arm has three revolute joints and one prismatic joint. The states $x_1, ..., x_4$ represent the joint angles/positions and $x_5, ..., x_8$ represent the joint velocities. The arm is fully actuated, and $u_1, ..., u_4$ are the input joint torques/forces. We consider a vertical force $|f_l(t)| \leq 0.025$~\si{N} applied at the end-effector as the disturbance. Let the state space be $\mathcal{X} = \{x \in \R^8 \mid |x_1| \leq 0.59, |x_2| \leq 0.51, |x_3| \leq 0.51, |x_4| \leq 0.50, |x_5| \leq 0.79, |x_6| \leq 1.39, |x_7| \leq 1.00, |x_8| \leq 0.71\}$ and control limits $\mathcal{U} = \{u \in \R^4 \mid |u_1| \leq 4.52, |u_2| \leq 3.06, |u_3| \leq 1.67, 3.62 \leq u_4 \leq 6.19 \}$. The system parameters are the link masses $m_1=2.0, m_2=1.0, m_3=m_4=0.5$~\si{kg}, moments of inertia $I_{z1}=I_{z2}=I_{z3}=I_{z4}=0.1$~\si{kg.m^2}, link lengths $l_1=l_2=0.4$~\si{m}, and distances to the center of mass $r_1=r_2=0.2$~\si{m}. We apply a change of variables $\tilde{u}_1 = u_1$, $\tilde{u}_2 = u_2$, $\tilde{u}_3 = u_3$, and $u_4 = m_4 g + \tilde{u}_4$ to shift the controls to zero at $x=0$.

We set $\mu=0.24$ for this example, and Fig.~\ref{fig:eg_scara_sim} displays 200 trajectories starting from the boundary of $\Omega_{\widehat{V}_r}$ under $f_l(t) = 0.025 \sin(0.5 \pi t)$. With the same initial conditions, the average control effort ($\int_0^T \lVert \tilde{u}(t) \rVert_2^2 dt$) of LQR and our method is 0.13 and 0.10, respectively. The average settling time ($\inf \{ \tau \geq 0 \mid \lVert x(t) \rVert_2 \leq 0.05, \forall t \geq \tau \}$) is 5.26~\si{s} and 5.56~\si{s}, respectively.

\section{Conclusion}\label{sec:conclusion}
This work presents a novel framework for learning robust control Lyapunov functions and synthesizing stabilizing controllers for continuous-time nonlinear systems under additive disturbances upper bounded by a function of state. We establish bounds on the Hessian and third-order derivatives of LNNs and propose a GPU-friendly branch-and-bound verification algorithm. This algorithm substantially reduces verification time compared to previous approaches that either rely solely on zeroth-order bounds or run on CPUs. Simulation studies on six distinct dynamical systems demonstrate the effectiveness and robustness of our approach. 

{\color{revcolor} While our results show significant advantages over prior approaches, we explicitly acknowledge that scalability to ultra-high-dimensional nonlinear systems (such as large-scale networked topologies) remains a fundamental challenge. Because the computational complexity of a BnB-based verifier grows rapidly with the state dimension, formal verification can eventually become intractable for arbitrary dimensions. Addressing this bottleneck remains an open research direction. Future work will focus on developing dimension-reduction strategies, exploiting decentralized system structures, and incorporating tighter relaxations to further enhance scalability and enable formal verification of increasingly complex systems.}

\section*{References}
\bibliographystyle{IEEEtran}
\bibliography{master}

\appendix

\subsection{Training Details}\label{appx:training_details}
The hyperparameters used for training and verification are listed in Table~\ref{tab:hyperparams}. Table~\ref{tab:nn_structure} presents the network structures and activation functions used for each example. 

{\color{revcolor} Our framework utilizes LQR pretraining to provide a crucial ``warm start'' that ensures local stability, as training from random initialization frequently fails for high-dimensional systems. Several existing works on neural Lyapunov functions also use the LQR solution as a prior \cite{chang2019neural, zhou2022neural, wu2023neural}. The prescribed Lipschitz bounds of the controllers ($\gamma_\pi$) are set to match or slightly exceed the spectral norm of the baseline LQR gain matrix to guarantee sufficient mathematical expressivity and control authority. In practical deployments, the sensitivity to measurement noise inherent to high-gain controllers can be effectively mitigated by routing raw measurements through standard state estimators or filters.}

\begin{table}[H]
\centering
\caption{Hyperparameters used for each example.}
\label{tab:hyperparams}
\begin{threeparttable}
\begin{tabular}{cccccc} 
\toprule
Example & $\gamma_V$ & $\gamma_\pi$ & $\eta$ & $\mu$ & $\omega(x)$ \\ 
\midrule
Inverted Pendulum       & 1.0 & 38.0 & 0.1  & 0.3  & $0.02\lVert x \rVert_2$ \\
Unicycle  & 1.0 & 3.0  & 0.05 & 0.2  & $0.001\lVert x \rVert_2$ \\
Third-Order System  & 5.0 & 12.5 & 0.5 & 0.5  & $0.001\lVert x \rVert_2$ \\
Cartpole                & 1.0 & 84.0 & 0.05 & 0.05 & $10^{-5}\lVert x \rVert_2$ \\
2D Quadrotor            & 1.0 & 24.0 & 0.1  & 0.12 & $10^{-6}\lVert x \rVert_2$ \\
SCARA Arm               & 1.0 & 3.0  & 0.1  & 0.24 & $10^{-6}\lVert x \rVert_2$ \\
\bottomrule
\end{tabular}
\begin{tablenotes}
\footnotesize
\item[a] $\delta_{\textrm{positive}} = \delta_{\textrm{decrease}} = \delta_{\textrm{inclusion}} = 10^{-6}$ for all examples.
\end{tablenotes}
\end{threeparttable}
\end{table}

\begin{table}[H]
\centering
\caption{Neural network structures and activations.}
\label{tab:nn_structure}
\setlength{\tabcolsep}{2.5pt} 
\begin{tabular}{c c c c c}
\toprule
\multirow{2}{*}{Example} & \multicolumn{3}{c}{$V_{\theta_1}$} & \textbf{$\pi_{\theta_2}$} \\
 \cmidrule(lr){2-4} \cmidrule(lr){5-5}
 & 2-layer & 3-layer & 4-layer & 2-layer \\ 
\midrule
\makecell{Inverted \\ Pendulum } 
    & \makecell{[2, 128, 1] \\ {[tanh, id]}} 
    & \makecell{[2, 128x2, 1] \\ {[tanh x2, id]}} 
    & \makecell{[2, 128x3, 1] \\ {[tanh x3, id]}} 
    & \makecell{[2, 32, 1] \\ {[ReLU, id]}} \\ \addlinespace
\makecell{Unicycle}
    & \makecell{[2, 128, 1] \\ {[tanh, id]}} 
    & \makecell{[2, 128x2, 1] \\ {[tanh x2, id]}} 
    & \makecell{[2, 128x3, 1] \\ {[tanh x3, id]}} 
    & \makecell{[2, 32, 1] \\ {[ReLU, id]}}  \\ \addlinespace
\makecell{Third-Order System}
    & \makecell{[3, 128, 1] \\ {[tanh, id]}} 
    & \makecell{[3, 128x2, 1] \\ {[tanh x2, id]}} 
    & \makecell{[3, 128x3, 1] \\ {[tanh x3, id]}}
    & \makecell{[3, 32, 1] \\ {[ReLU, id]}} \\ \addlinespace
\makecell{Cartpole} 
    & \makecell{[4, 128, 1] \\ {[tanh, id]}} 
    & \makecell{[4, 128x2, 1] \\ {[tanh x2, id]}} 
    & \makecell{[4, 128x3, 1] \\ {[tanh x3, id]}}
    & \makecell{[4, 32, 1] \\ {[ReLU, id]}}\\ \addlinespace
\makecell{2D Quadrotor} 
    & \makecell{[6, 128, 1] \\ {[tanh, id]}} 
    & -- & -- 
    & \makecell{[6, 32, 2] \\ {[ReLU, id]}} \\ \addlinespace
\makecell{SCARA Arm} 
    & \makecell{[8, 128, 1] \\ {[tanh, id]}} 
    & -- & -- 
    & \makecell{[8, 32, 4] \\ {[ReLU, id]}} \\
\bottomrule
\end{tabular}
\end{table}

\subsection{Proof of Lemma~\ref{lemma:lip_nn_hess}}\label{appx:proof_lip_nn_hess}
{\color{revcolor}
Consider the LNN $F(x)$ defined in \eqref{eq:lip_nn}. By Lemma~\ref{lemma:lip_nn_jac}, the Jacobian of $F_i(x)$ (the $i$-th component of $F(x)$) w.r.t. $x$ is 
$\frac{\partial F_i(x)}{\partial x} = \frac{\partial F_i(x)}{\partial h^{(L-1)}} \frac{\partial h^{(L-1)}}{\partial h^{(L-2)}} \cdots \frac{\partial h^{(1)}}{\partial h^{(0)}} \frac{\partial h^{(0)}}{\partial x}$. Further differentiate $\frac{\partial F_i(x)}{\partial x}$ w.r.t. $x_j$ and we have 
$\frac{\partial^2 F_i(x)}{\partial x_j \partial x} = \sum_{\ell = 1}^{L-1} \frac{\partial F_i(x)}{\partial h^{(\ell)}} \frac{\partial^2 h^{(\ell)}}{\partial x_j \partial h^{(\ell-1)}} \frac{\partial h^{(\ell-1)}}{\partial x}$ 
because $\frac{\partial F_i(x)}{\partial h^{(L-1)}} \allowbreak = \sqrt{\gamma} W^{(L)}_i$ and $\frac{\partial h^{(0)}}{\partial x} = \sqrt{\gamma} I$ are both independent of $x$. For $\ell = 1,..., L-1$, we get from Lemma~\ref{lemma:lip_nn_jac} that 
$\frac{\partial^2 h^{(\ell)}}{\partial x_j \partial h^{(\ell-1)}} = U^{(\ell)} \diag \big ( {\sigma^{(\ell)}}''(z^{(\ell)}) \odot \frac{\partial z^{(\ell)}}{\partial x_j} \big ) W^{(\ell)}$. Next, plug $\frac{\partial^2 h^{(\ell)}}{\partial x_j \partial h^{(\ell-1)}}$ into $\frac{\partial^2 F_i(x)}{\partial x_j \partial x}$ and we get $\frac{\partial^2 F_i(x)}{\partial x_j \partial x} = \sum_{\ell=1}^{L-1} \frac{\partial F_i(x)}{\partial h^{(\ell)}} U^{(\ell)} \diag \big ( {\sigma^{(\ell)}}''(z^{(\ell)}) \odot \frac{\partial z^{(\ell)}}{\partial x_j} \big ) W^{(\ell)} \allowbreak \frac{\partial h^{(\ell-1)}}{\partial x} = \sum_{\ell=1}^{L-1} (\frac{\partial z^{(\ell)}}{\partial x_j})^\top \diag \big ({\sigma^{(\ell)}}''(z^{(\ell)}) \odot   (\frac{\partial F_i(x)}{\partial h^{(\ell)}} U^{(\ell)} )^\top \big ) \allowbreak W^{(\ell)} \frac{\partial h^{(\ell-1)}}{\partial x}$
where we use the fact that $a \diag (b \odot c) = c^\top \diag (b \odot a^\top)$ for a row vector $a$ and two column vectors $b$ and $c$ of the same length. Finally, noting that $\frac{\partial z^{(\ell)}}{\partial x} = W^{(\ell)} \frac{\partial h^{(\ell-1)}}{\partial x}$ for $\ell = 1, ..., L-1$ and stacking $\frac{\partial^2 F_i(x)}{\partial x_j \partial x} \in \R^{1 \times n_x}$ vertically, we obtain the desired result.}

\subsection{Proof of Theorem~\ref{thm:lip_nn_third_order_bound}}\label{appx:proof_lip_nn_third_order_bound}
{\color{revcolor} 
Before proving Theorem~\ref{thm:lip_nn_third_order_bound}, we establish the bounds on the forward Hessian $ \big\lVert \frac{\partial^2 h^{(\ell)}}{\partial x_j \partial x} \big \rVert_2$ for $\ell = 0, \dots, L-2$, and the backward Hessian $\big \lVert \frac{\partial^2 F_i(x)}{\partial x_j \partial h^{(\ell)}} \big \rVert_2$ for $\ell = 1, \dots, L-1$. 

\textbf{Forward Bound \eqref{eq:fwd_bound}.} When $\ell = 0$, the bound is 0 since $h^{(0)} = \sqrt{\gamma} x$. For $\ell \ge 1$, we expand using the chain rule: $\frac{\partial^2 h^{(\ell)}}{\partial x_j \partial x} = \sum_{k=1}^\ell \frac{\partial h^{(\ell)}}{\partial h^{(k)}} \frac{\partial^2 h^{(k)}}{\partial x_j \partial h^{(k-1)}} \frac{\partial h^{(k-1)}}{\partial x}$. We know the local Hessian is $\frac{\partial^2 h^{(k)}}{\partial x_j \partial h^{(k-1)}} = U^{(k)} \diag \bigl(\frac{\partial z^{(k)}}{\partial x_j} \odot {\sigma^{(k)}}''(z^{(k)})\bigr) W^{(k)}$.
Applying Cauchy–Schwarz inequality, its spectral norm is bounded by $\lVert c(U^{(k)}) \rVert_\infty \big \lVert \frac{\partial z^{(k)}}{\partial x_j} \big \rVert_2 \sup|{\sigma^{(k)}}''| \lVert W^{(k)} \rVert_2$. Because $\big \lVert \frac{\partial z^{(k)}}{\partial x_j} \big \rVert_2 \le \lVert W^{(k)} \rVert_2 \big \lVert \frac{\partial h^{(k-1)}}{\partial x_j} \big \rVert_2 \le \sqrt{\gamma} \lVert W^{(k)} \rVert_2$, we have $ \big \lVert \frac{\partial^2 h^{(k)}}{\partial x_j \partial h^{(k-1)}}  \big \rVert_2 \le \sqrt{\gamma} \lVert c(U^{(k)}) \rVert_\infty \lVert W^{(k)} \rVert_2^2 \sup|{\sigma^{(k)}}''|$.
As $\lVert \frac{\partial h^{(\ell)}}{\partial h^{(k)}} \rVert_2 \le 1$ and $\lVert \frac{\partial h^{(k-1)}}{\partial x} \rVert_2 \le \sqrt{\gamma}$, the summation simplifies to \eqref{eq:fwd_bound}.

\textbf{Backward Bound \eqref{eq:bwd_bound}.} When $\ell = L-1$, the bound is 0 because $F(x) = \sqrt{\gamma} W^{(L)} h^{(L-1)}$ is purely linear with respect to $h^{(L-1)}$. For $\ell \le L-2$, we have 
$\frac{\partial^2 F_i(x)}{\partial x_j \partial h^{(\ell)}} = \sum_{k=\ell+1}^{L-1} \frac{\partial F_i(x)}{\partial h^{(k)}} \frac{\partial^2 h^{(k)}}{\partial x_j \partial h^{(k-1)}} \frac{\partial h^{(k-1)}}{\partial h^{(\ell)}}$.
Taking the spectral norm and applying sub-multiplicativity yields: $\big \lVert \frac{\partial^2 F_i(x)}{\partial x_j \partial h^{(\ell)}} \big \rVert_2 \leq \sum_{k=\ell+1}^{L-1} \big \lVert \frac{\partial F_i(x)}{\partial h^{(k)}} \big \rVert_2 \big \lVert \frac{\partial^2 h^{(k)}}{\partial x_j \partial h^{(k-1)}} \big \rVert_2 \big \lVert \frac{\partial h^{(k-1)}}{\partial h^{(\ell)}} \big \rVert_2$. We already established that the local Hessian norm is bounded by $\sqrt{\gamma} \lVert c(U^{(k)}) \rVert_\infty \lVert W^{(k)} \rVert_2^2 \sup|{\sigma^{(k)}}''|$. Furthermore, due to the 1-Lipschitz property of the intermediate layers, $\big \lVert \frac{\partial h^{(k-1)}}{\partial h^{(\ell)}} \big \rVert_2 \leq 1$ and $\big \lVert \frac{\partial F_i(x)}{\partial h^{(k)}} \big \rVert_2 \leq \sqrt{\gamma}$. Substituting these three bounds directly into the summation yields the backward bound \eqref{eq:bwd_bound}.

\textbf{Bound on $\big \lVert \frac{\partial^3 F_i(x)}{\partial x^2 \partial x_j} \big \rVert_2$.} With $\Theta^{(\ell)} = W^{(\ell)} \frac{\partial h^{(\ell-1)}}{\partial x}$ and $\Lambda^{(\ell)} = \diag(\lambda^{(\ell)})$, $\lambda^{(\ell)} = (\frac{\partial F_i(x)}{\partial h^{(\ell)}} U^{(\ell)})^\top \odot {\sigma^{(\ell)}}''(z^{(\ell)})$, we have 
$
\frac{\partial^3 F_i(x)}{\partial x^2 \partial x_j} = \sum_{\ell=1}^{L-1}  \Theta^{(\ell)\top} \Lambda^{(\ell)} \tfrac{\partial \Theta^{(\ell)}}{\partial x_j} + \tfrac{\partial \Theta^{(\ell) \top}}{\partial x_j} \Lambda^{(\ell)} \Theta^{(\ell)}  + \Theta^{(\ell)\top} \diag(p^{(\ell)}) \Theta^{(\ell)} + \Theta^{(\ell)\top} \diag(q^{(\ell)}) \Theta^{(\ell)},
$ where $p^{(\ell)} = (\frac{\partial^2 F_i(x)}{\partial x_j \partial h^{(\ell)}} U^{(\ell)})^\top \odot {\sigma^{(\ell)}}''(z^{(\ell)})$ and $q^{(\ell)} = (\frac{\partial F_i(x)}{\partial h^{(\ell)}} U^{(\ell)})^\top \odot {\sigma^{(\ell)}}'''(z^{(\ell)}) \odot \frac{\partial z^{(\ell)}}{\partial x_j}$ are the two summands of $\frac{\partial \Lambda^{(\ell)}}{\partial x_j}$.

For any nonnegative weight vector $w$ (indexed by the hidden units of layer $\ell$), let $\Gamma_w := \lVert \Theta^{(\ell)\top} \diag(w) \Theta^{(\ell)} \rVert_2 = \lVert \diag(w)^{1/2} \Theta^{(\ell)} \rVert_2^2$. Using $\big \lVert \frac{\partial h^{(\ell-1)}}{\partial x} \big \rVert_2 \le \sqrt{\gamma}$, we have $\Gamma_{c(U^{(\ell)})} \le \gamma\, \beta^{(\ell)}$ and $\Gamma_{c(U^{(\ell)}) \odot r(W^{(\ell)})} \le \gamma\, \tau^{(\ell)}$.

\textbf{Term $\Theta^{(\ell)\top} \diag(p^{(\ell)}) \Theta^{(\ell)} $.} As in the proof of Theorem~\ref{thm:lip_nn_hess_bound}, $\lVert \Theta^{(\ell)\top} \diag(p^{(\ell)}) \Theta^{(\ell)} \rVert_2 \le \Gamma_{|p^{(\ell)}|}$. By Cauchy--Schwarz, $|p^{(\ell)}_k| \le \sup|{\sigma^{(\ell)}}''| \big \lVert \frac{\partial^2 F_i(x)}{\partial x_j \partial h^{(\ell)}} \big \rVert_2 [c(U^{(\ell)})]_k$, so $\lVert \Theta^{(\ell)\top} \diag(p^{(\ell)}) \Theta^{(\ell)}  \rVert_2  \le \gamma\, \beta^{(\ell)} \sup|{\sigma^{(\ell)}}''| \big \lVert \frac{\partial^2 F_i(x)}{\partial x_j \partial h^{(\ell)}} \big \rVert_2$.

\textbf{Term $\Theta^{(\ell)\top} \diag(q^{(\ell)}) \Theta^{(\ell)}$.} Since $\frac{\partial z^{(\ell)}}{\partial x_j} = W^{(\ell)} \frac{\partial h^{(\ell-1)}}{\partial x_j}$ and $\big \lVert \frac{\partial h^{(\ell-1)}}{\partial x_j} \big \rVert_2 \le \sqrt{\gamma}$, we have $\bigl|(\frac{\partial z^{(\ell)}}{\partial x_j})_k\bigr| \le \sqrt{\gamma}\,[r(W^{(\ell)})]_k$, while $\bigl|(\frac{\partial F_i(x)}{\partial h^{(\ell)}} U^{(\ell)})_k\bigr| \le \sqrt{\gamma}\,[c(U^{(\ell)})]_k$ as before. Hence $|q^{(\ell)}_k| \le \gamma \sup|{\sigma^{(\ell)}}'''| [c(U^{(\ell)}) \odot r(W^{(\ell)})]_k$, and $\lVert \Theta^{(\ell)\top} \diag(q^{(\ell)}) \Theta^{(\ell)} \rVert_2 \le \Gamma_{|q^{(\ell)}|} \le \gamma \sup|{\sigma^{(\ell)}}'''| \, \Gamma_{c(U^{(\ell)}) \odot r(W^{(\ell)})} \le \gamma^2 \sup|{\sigma^{(\ell)}}'''| \, \tau^{(\ell)}$.

\textbf{Term $\Theta^{(\ell)\top} \Lambda^{(\ell)} \tfrac{\partial \Theta^{(\ell)}}{\partial x_j}$.} 
$\bigl\lVert \Theta^{(\ell)\top} \Lambda^{(\ell)} \tfrac{\partial \Theta^{(\ell)}}{\partial x_j} \bigr\rVert_2 \le  \bigl\lVert \diag(|\lambda^{(\ell)}|)^{1/2} \Theta^{(\ell)} \bigr\rVert_2 \, \bigl\lVert \diag(|\lambda^{(\ell)}|)^{1/2} \tfrac{\partial \Theta^{(\ell)}}{\partial x_j} \bigr\rVert_2$, as  $\Lambda^{(\ell)}$ is diagonal. As $|\lambda^{(\ell)}_k| \le \allowbreak \sqrt{\gamma}\sup|{\sigma^{(\ell)}}''| [c(U^{(\ell)})]_k$, the first factor obeys $\lVert \diag(|\lambda^{(\ell)}|)^{1/2} \Theta^{(\ell)} \rVert_2^2 \le \sqrt{\gamma}\sup|{\sigma^{(\ell)}}''| \, \Gamma_{c(U^{(\ell)})} \le \gamma^{3/2}\sup|{\sigma^{(\ell)}}''| \beta^{(\ell)}$. For the second, $\frac{\partial \Theta^{(\ell)}}{\partial x_j} = W^{(\ell)} \frac{\partial^2 h^{(\ell-1)}}{\partial x_j \partial x}$, and it follows that $\bigl\lVert \diag(|\lambda^{(\ell)}|)^{1/2} \tfrac{\partial \Theta^{(\ell)}}{\partial x_j} \bigr\rVert_2^2 \le \sqrt{\gamma}\sup|{\sigma^{(\ell)}}''| \allowbreak  \bigl\lVert \diag(c(U^{(\ell)}))^{1/2} W^{(\ell)} \bigr\rVert_2^2 \allowbreak  \bigl\lVert \tfrac{\partial^2 h^{(\ell-1)}}{\partial x_j \partial x} \bigr\rVert_2^2 = \sqrt{\gamma}\sup|{\sigma^{(\ell)}}''| \allowbreak  \beta^{(\ell)} \bigl\lVert \tfrac{\partial^2 h^{(\ell-1)}}{\partial x_j \partial x} \bigr\rVert_2^2 $.
Taking the square root, $\lVert  \Theta^{(\ell)\top} \Lambda^{(\ell)} \tfrac{\partial \Theta^{(\ell)}}{\partial x_j} \rVert_2  \allowbreak \le \gamma\, \beta^{(\ell)} \sup|{\sigma^{(\ell)}}''| \allowbreak  \lVert \frac{\partial^2 h^{(\ell-1)}}{\partial x_j \partial x} \rVert_2$.

As $\lVert \Theta^{(\ell)\top} \Lambda^{(\ell)} \tfrac{\partial \Theta^{(\ell)}}{\partial x_j} \rVert_2  = \lVert \tfrac{\partial \Theta^{(\ell) \top}}{\partial x_j} \Lambda^{(\ell)} \Theta^{(\ell)} \rVert_2$, summing these terms over $\ell = 1, \dots, L-1$ yields \eqref{eq:lip_nn_third_order_bound}.
}

\subsection{High-order Lower Bound for $\Phi(x) = -H(x)-\omega(x)$}\label{appx:high_order_bound_stability}

The target function for verifying the decrease condition is $\Phi(x) = -H(x) - \omega(x)$, where $ H(x) \defeq \frac{\partial V_{\theta_1}}{\partial x} (x) f(x,\pi_{\theta_2}(x)) + \lVert \frac{\partial V_{\theta_1}}{\partial x} (x) G  \rVert_2 \epsilon(x) $.
Because $\pi_{\theta_2}$ and the $\ell_2$-norm operator are only $\mathcal{C}^0$, a direct Taylor expansion of $H(x)$ is impossible. Instead, we decompose $H(x) = H_1(x) + H_2(x)$ where $H_1 (x) = \frac{\partial V_{\theta_1}}{\partial x}(x) f(x,\pi_{\theta_2}(x))$ and $ H_2 (x) = \lVert \frac{\partial V_{\theta_1}}{\partial x}(x) G \rVert_2 \epsilon(x) $.
Given a rectangle $Q$, let $x_m$ be its midpoint, $v$ be its half-width vector, $\zeta = \lVert v \rVert_2$, and $U_Q$ the hyperrectangle formed by the elementwise lower and upper bounds of $\pi_{\theta_2}$ on $Q$. To rigorously lower bound $\Phi(x)$, we will derive upper bounds for $H_1(x)$, $H_2(x)$, and $\omega(x)$ on $Q$.

\textbf{Taylor Expansions.}
Because $f$ is $\mathcal{C}^2$, we can expand the system dynamics around $x_m$. For any $x \in Q$ there exists $\lambda_i' \in (0, 1)$ defining the intermediate points $x_i' = x_m + \lambda_i'(x - x_m)$ and $u_i' = \pi_{\theta_2}(x_m) + \lambda_i'(\pi_{\theta_2}(x) - \pi_{\theta_2}(x_m))$ such that $f(x,\pi_{\theta_2}(x)) = f(x_m,\pi_{\theta_2}(x_m)) + \frac{\partial f}{\partial x} (x_m, \pi_{\theta_2}(x_m)) (x-x_m) + \frac{\partial f}{\partial u} (x_m, \pi_{\theta_2}(x_m)) (\pi_{\theta_2}(x)-\pi_{\theta_2}(x_m)) +  [ \frac{1}{2} (x-x_m)^\top \frac{\partial^2 f_i}{\partial x^2} (x_i', u_i') (x-x_m)  ]_i +  [ \frac{1}{2} (x-x_m)^\top \frac{\partial^2 f_i}{\partial u \partial x} (x_i', u_i') (\pi_{\theta_2}(x)-\pi_{\theta_2}(x_m))  ]_i +  [ \frac{1}{2} (\pi_{\theta_2}(x)-\pi_{\theta_2}(x_m))^\top \frac{\partial^2 f_i}{\partial x \partial u} (x_i', u_i') (x-x_m)  ]_i  +  [ \frac{1}{2} (\pi_{\theta_2}(x)-\pi_{\theta_2}(x_m))^\top \frac{\partial^2 f_i}{\partial u^2} (x_i', u_i') (\pi_{\theta_2}(x)-\pi_{\theta_2}(x_m))  ]_i$.

Similarly, because $V_{\theta_1}$ is $\mathcal{C}^3$, for $\lambda_i'' \in (0,1)$ and intermediate points $x_i'' = x_m + \lambda_i''(x - x_m)$,
$\frac{\partial V_{\theta_1}}{\partial x}(x) = \frac{\partial V_{\theta_1}}{\partial x}(x_m) + (x-x_m)^\top \frac{\partial^2 V_{\theta_1}}{\partial x^2}(x_m) +  [ \frac{1}{2} (x-x_m)^\top \frac{\partial^3 V_{\theta_1}}{\partial x^2 \partial x_i}(x_i'') (x-x_m)  ]_i$.

\textbf{Bounding $H_1(x)$.}
By multiplying the expansions above and bounding the higher-order remainder terms using the maximum spectral norms over the domain, we obtain an upper bound for $H_1(x)$: $H_1 (x) 
\leq H_1(x_m) +  \lvert \frac{\partial V_{\theta_1}}{\partial x}(x_m)  \allowbreak   \frac{\partial f}{\partial x} (x_m, \pi_{\theta_2}(x_m)) + f^\top(x_m, \pi_{\theta_2}(x_m)) \frac{\partial^2 V_{\theta_1}}{\partial x^2}(x_m)  \rvert v  
+  \lVert \frac{\partial V_{\theta_1}}{\partial x}(x_m) \frac{\partial f}{\partial u} (x_m, \pi_{\theta_2}(x_m))  \rVert_2 \gamma_\pi \zeta + \frac{1}{2}  \lvert \frac{\partial V_{\theta_1}}{\partial x}(x_m)  \rvert h_f \zeta^2 
+  v^\top \lvert \frac{\partial^2 V_{\theta_1}}{\partial x^2}(x_m)  \frac{\partial f}{\partial x} (x_m, \pi_{\theta_2}(x_m))  \rvert v 
+  \lVert v^\top \lvert \frac{\partial^2 V_{\theta_1}}{\partial x^2}(x_m)   \allowbreak   \frac{\partial f}{\partial u} (x_m, \pi_{\theta_2}(x_m)) \rvert \rVert_2 \gamma_\pi \zeta
+ \frac{1}{2} t_V ^\top \lvert f(x_m, \pi_{\theta_2}(x_m))  \rvert \zeta^2 
+ \frac{1}{2}  v^\top  \lvert \frac{\partial^2 V_{\theta_1}}{\partial x^2}(x_m)  \rvert h_f   \zeta^2 
+ \frac{1}{2}  t_V^\top \lvert \frac{\partial f}{\partial x} (x_m, \pi_{\theta_2}(x_m))   \rvert v \zeta^2 
\allowbreak 
+ \frac{1}{2}  \lVert t_V^\top \lvert \frac{\partial f}{\partial u} (x_m, \pi_{\theta_2}(x_m))  \rvert  \rVert_2 \gamma_\pi \zeta^3 
+ \frac{1}{4}  t_V^\top h_f  \zeta^4
$ where $\gamma_\pi \geq \Lip(\pi_{\theta_2})$, and the elements of the vectors $h_f$ and $t_V$ are defined as: $h_{f,i} = 2 \gamma_\pi \max_{x\in Q, u \in U_Q} \lVert \frac{\partial^2 f_i}{\partial x \partial u} (x,u) \rVert_2 + \gamma_\pi^2 \max_{x\in Q, u \in U_Q} \lVert \frac{\partial^2 f_i}{\partial u^2} (x,u) \rVert_2 + \max_{x\in Q, u \in U_Q} \lVert \frac{\partial^2 f_i}{\partial x^2}(x,u) \rVert_2$ and $t_{V,i} = T_V$ with $T_V \geq \sup_{x \in \R^{n_x}} \lVert \frac{\partial^3 V_{\theta_1}}{\partial x^2 \partial x_j} (x) \rVert_2$ for all $j$.

\textbf{Bounding $H_2(x)$ and $\omega(x)$.}
Applying similar expansions and norm inequalities to the disturbance term $H_2(x)$, we have: $H_2(x) \leq \lVert \lvert \frac{\partial V_{\theta_1}}{\partial x} (x_m) G\rvert + v^\top \lvert \frac{\partial^2 V_{\theta_1}}{\partial x^2} (x_m) G\rvert + \frac{1}{2} t_V^\top \lvert G \rvert \zeta^2 \rVert_2 ( \epsilon(x_m) + K_\epsilon \zeta ) $ where $K_\epsilon \geq \Lip (\epsilon)$. Finally, the function $\omega(x)$ is upper bounded by $\omega(x) \leq \omega(x_m) + K_\omega \zeta$ where $K_\omega \geq \Lip (\omega)$. 

By combining the upper bounds of $H_1$, $H_2$, and $\omega$, we obtain the required high-order lower bound for the verification target $\min_{x \in Q} \Phi(x)$.

\subsection{Comparison between the Zeroth- and High-Order Bounds}\label{appx:comparison_between_bounds}
{\color{revcolor} Table~\ref{tab:comparison_between_bounds} compares the verification times of Algorithm~\ref{alg:bnb} when using only the zeroth-order bounds versus the proposed higher-order bounds with the Sandwich Layers \cite{wang2023direct}. While the algorithm is theoretically capable of running solely with the zeroth-order bounds, doing so quickly becomes computationally intractable. As demonstrated empirically, the zeroth-order baseline consistently times out on the positive definiteness and decrease conditions for systems with state dimensions $n_x \geq 4$, highlighting the necessity of the higher-order bounds.}

\begin{table}[H]
\centering
\footnotesize
\setlength{\tabcolsep}{3.5pt} %
\caption{{\color{revcolor}Verification times (in seconds unless otherwise specified) for each system using different bounds on GPU (average over four trials). ZO: zeroth-order. HO: high-order.}}
\label{tab:comparison_between_bounds}
{\color{revcolor}
\begin{tabular}{c c cc cc cc}
\toprule
\multirow{2}{*}{System} & \multirow{2}{*}{$L_V$} & \multicolumn{2}{c}{Pos. Def.} & \multicolumn{2}{c}{Decrease} & \multicolumn{2}{c}{Inclusion} \\
\cmidrule(lr){3-4} \cmidrule(lr){5-6} \cmidrule(lr){7-8} & & ZO & HO & ZO & HO & ZO & HO \\
\midrule
\multirow{3}{*}{\makecell[c]{Inverted\\Pendulum}}  
& 2 & \textbf{0.16} & 0.31 & \textbf{0.17} & 0.71 & -- & -- \\
& 3 & \textbf{0.17} & 0.30 & \textbf{0.18} & 0.71 & -- & -- \\
& 4 & \textbf{0.16} & 0.33 & \textbf{0.19} & 0.62 & -- & -- \\
\midrule
\multirow{3}{*}{\makecell[c]{Cartpole}}  
& 2 & $>$ 4~\si{h} & \textbf{0.35} & $>$ 4~\si{h} & \textbf{13.28} & 0.45 & \textbf{0.13} \\
& 3 & $>$ 4~\si{h} & \textbf{0.39} & $>$ 4~\si{h} & \textbf{29.24} & 0.46 & \textbf{0.13} \\
& 4 & $>$ 4~\si{h} & \textbf{0.51} & $>$ 4~\si{h} & \textbf{78.25} & 0.46 & \textbf{0.12} \\
\midrule
SCARA Arm  
& 2 & $>$ 12~\si{h} & \textbf{9.05} & $>$ 12~\si{h} & \textbf{9.53~\si{h}} & 5.02~\si{h} & \textbf{7.67} \\
\bottomrule
\end{tabular}}
\end{table}

\subsection{Additional Results with the Spectral Normalization and Orthogonal Layer Architectures}\label{appx:additional_results_on_other_layers}
{\color{revcolor} In Tables~\ref{tab:spectral_normalization_times} and \ref{tab:orthogonal_layer_times}, we report the performance of the BnB algorithm on the inverted pendulum and the cartpole when integrated with the Spectral Normalization \cite{miyato2018spectral} and Orthogonal Layer \cite{trockman2021orthogonalizing} architectures, respectively. These results demonstrate that the computational speedups achieved by our BnB algorithm are consistent and independent of the specific underlying Lipschitz parameterization.

Furthermore, our empirical evaluations indicate that the Sandwich Layer architecture \cite{wang2023direct} offers superior expressivity compared to the other two architectures. For instance, in the inverted pendulum example, the average verified forward invariant region covers 59.52\% of the state space when using Spectral Normalization and 61.31\% with Orthogonal Layers, but increases significantly to 73.75\% when using Sandwich Layers. Similarly, for the cartpole system, the invariant region volumes are 5.34\%, 5.30\%, and 5.40\% for Spectral Normalization, Orthogonal Layer, and Sandwich Layer, respectively.}

\begin{table}[H]
\centering
\footnotesize
\setlength{\tabcolsep}{3.5pt} %
\caption{{\color{revcolor}Verification times (in seconds) for the Spectral Normalization architecture on GPU (average over four trials).}}
\label{tab:spectral_normalization_times}
{\color{revcolor}
\begin{tabular}{c c cc cc cc}
\toprule
\multirow{3}{*}{System} & \multirow{3}{*}{$L_V$} & \multicolumn{2}{c}{Pos. Def.} & \multicolumn{2}{c}{Decrease} & \multicolumn{2}{c}{Inclusion} \\
\cmidrule(lr){3-4} \cmidrule(lr){5-6} \cmidrule(lr){7-8} & & \multirow{2}{*}{\makecell{$\alpha,\beta$-\\CROWN}} & \multirow{2}{*}{\makecell{Ours}} & \multirow{2}{*}{\makecell{$\alpha,\beta$-\\CROWN}} & \multirow{2}{*}{\makecell{Ours}} & \multirow{2}{*}{\makecell{$\alpha,\beta$-\\CROWN}} & \multirow{2}{*}{\makecell{Ours}}  \\
& & & & & & & \\
\midrule
\multirow{3}{*}{\makecell[c]{Inverted\\Pendulum}}  
& 2 & 0.92 & \textbf{0.34} & 4.56 & \textbf{0.98} & -- & -- \\
& 3 & 1.24 & \textbf{0.30} & 6.67 & \textbf{1.03} & -- & -- \\
& 4 & 1.64 & \textbf{0.33} & 8.82 & \textbf{1.02} & -- & -- \\
\midrule
\multirow{3}{*}{\makecell[c]{Cartpole}}  
& 2 & 2.58 & \textbf{0.35} & NaN & \textbf{12.45} & 0.39 & \textbf{0.14} \\
& 3 & 4.63 & \textbf{0.38} & NaN & \textbf{22.56} & 0.47 & \textbf{0.13} \\
& 4 & 7.69 & \textbf{0.47} & NaN & \textbf{40.59} & 0.51 & \textbf{0.13} \\
\bottomrule
\end{tabular}}
\end{table}

\begin{table}[H]
\centering
\footnotesize
\setlength{\tabcolsep}{3.5pt} %
\caption{{\color{revcolor}Verification times (in seconds) for the Orthogonal Layer architecture on GPU (average over four trials).}}
\label{tab:orthogonal_layer_times}
{\color{revcolor}
\begin{tabular}{c c cc cc cc}
\toprule
\multirow{3}{*}{System} & \multirow{3}{*}{$L_V$} & \multicolumn{2}{c}{Pos. Def.} & \multicolumn{2}{c}{Decrease} & \multicolumn{2}{c}{Inclusion} \\
\cmidrule(lr){3-4} \cmidrule(lr){5-6} \cmidrule(lr){7-8} & & \multirow{2}{*}{\makecell{$\alpha,\beta$-\\CROWN}} & \multirow{2}{*}{\makecell{Ours}} & \multirow{2}{*}{\makecell{$\alpha,\beta$-\\CROWN}} & \multirow{2}{*}{\makecell{Ours}} & \multirow{2}{*}{\makecell{$\alpha,\beta$-\\CROWN}} & \multirow{2}{*}{\makecell{Ours}}  \\
& & & & & & & \\
\midrule
\multirow{3}{*}{\makecell[c]{Inverted\\Pendulum}}  
& 2 & 1.03 & \textbf{0.33} & 4.68 & \textbf{1.01} & -- & -- \\
& 3 & 1.26 & \textbf{0.29} & 6.72 & \textbf{0.98} & -- & -- \\
& 4 & 1.61 & \textbf{0.29} & 8.77 & \textbf{0.99} & -- & -- \\
\midrule
\multirow{3}{*}{\makecell[c]{Cartpole}}  
& 2 & 2.57 & \textbf{0.36} & NaN & \textbf{13.04} & 0.40 & \textbf{0.14} \\
& 3 & 4.64 & \textbf{0.41} & NaN & \textbf{25.66} & 0.46 & \textbf{0.13} \\
& 4 & 7.72 & \textbf{0.48} & NaN & \textbf{58.87} & 0.48 & \textbf{0.13} \\
\bottomrule
\end{tabular}}
\end{table}

\subsection{Two-Link Cartesian Arm}
\begin{figure}[H]
    \centering
    \includegraphics[width=0.4\linewidth]{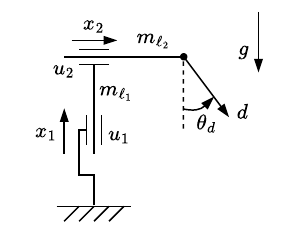}
    \caption{Diagram of the two-link Cartesian arm.}
    \label{fig:cartesian_arm_diagram}
\end{figure}

{\color{arxivcolor}
Consider the two-link Cartesian arm depicted in Fig.~\ref{fig:cartesian_arm_diagram}. The dynamics of this system are given by
\begin{subequations}
    \begin{align}
        \dot{x}_1 &= x_3, \quad \dot{x}_2 = x_4 \\
        \dot{x}_3 &= -g + (u_1 - \cos(\theta_d(t)) d)/(m_{\ell_1} + m_{\ell_2}) \\
        \dot{x}_4 &= (u_2 + \sin(\theta_d(t)) d)/m_{\ell_2}
    \end{align}
\end{subequations}
where $x_1$ and $x_2$ are the joint coordinates, $x_3$ and $x_4$ are the joint velocities, and $u_1$ and $u_2$ are the joint torques controlling the two joints, respectively. Let $m_{\ell_1}$ and $m_{\ell_2}$ be the masses of the two links, and assume that there is an external force $d$ acting on the end of the second link, forming a time-varying angle $\theta_d(t)$ with the downward direction. The system parameters are $m_{\ell_1} = m_{\ell_2} = 2$~\si{kg}, $g = 9.81$~\si{m.s^{-2}}, and the external force is $d=0.1$~\si{N}. Let the state space be $\mathcal{X} = \{x \in \R^4 \mid |x_1| \leq 0.5, |x_2| \leq 0.5, |x_3| \leq 0.5, |x_4| \leq 0.5\}$ and control limits $\mathcal{U} = \{u \in \R^2 \mid 32.24 \leq u_1 \leq 46.24, |u_2| \leq 3\}$. We apply a change of variables $u_1 = (m_{\ell_1} + m_{\ell_2})g + \tilde{u}_1$ and $u_2 = \tilde{u}_2$ to shift the controls to zero at $x=0$. Note that $\bigl \lVert \bigl [\frac{1}{m_{\ell_1} + m_{\ell_2}} \cos(\theta_d(t)) d, \frac{1}{ m_{\ell_2}} \sin(\theta_d(t)) d  \bigr ]^\top \bigr \rVert_2 \leq 0.05$.}

\begin{figure}[H]
    \centering
    \begin{subfigure}[b]{0.20\textwidth}
        \centering
        \includegraphics[width=\textwidth]{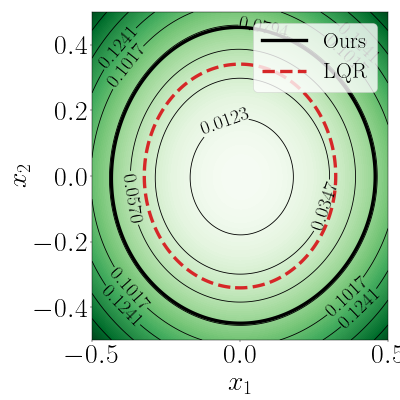}
        \caption{$V_{\theta_1}(x)$ on $(x_1, x_2)$}
        \label{fig:eg_cartesian_arm_forward_inv_1}
    \end{subfigure}
    \begin{subfigure}[b]{0.20\textwidth}
        \centering
        \includegraphics[width=\textwidth]{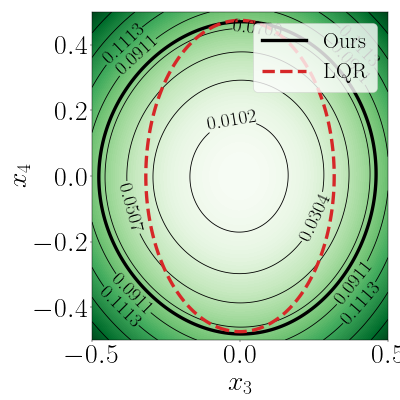}
        \caption{$V_{\theta_1}(x)$ on $(x_3, x_4)$}
        \label{fig:eg_cartesian_arm_forward_inv_2}
    \end{subfigure}
    \begin{subfigure}[b]{0.20\textwidth}
        \centering
        \includegraphics[width=\textwidth]{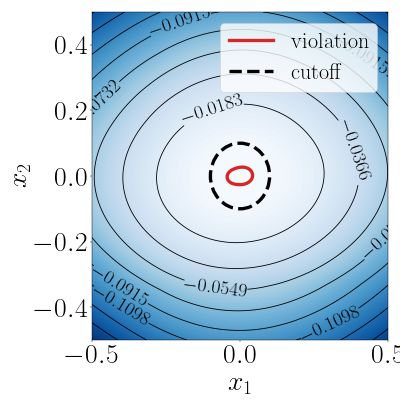}
        \caption{$H(x) + \omega (x)$ on $(x_1, x_2)$}
        \label{fig:eg_cartesian_arm_stability_1}
    \end{subfigure}
    \begin{subfigure}[b]{0.20\textwidth}
        \centering
        \includegraphics[width=\textwidth]{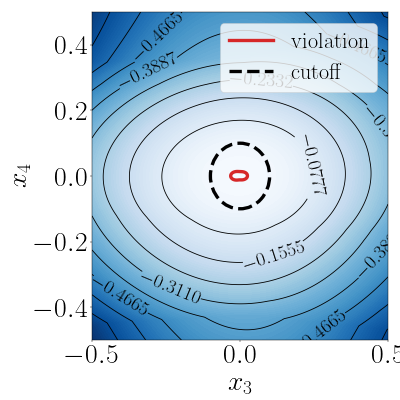}
        \caption{$H(x) + \omega (x)$ on $(x_3, x_4)$}
        \label{fig:eg_cartesian_arm_stability_2}
    \end{subfigure}
    \caption{{\color{arxivcolor} Visualization of $V_{\theta_1}(x)$ and $H(x) + \omega (x)$ for the two-link Cartesian arm. \subref{fig:eg_cartesian_arm_forward_inv_1} and \subref{fig:eg_cartesian_arm_forward_inv_2}: level sets of $V_{\theta_1}(x)$ and the largest forward invariant sets given by our method and LQR. \subref{fig:eg_cartesian_arm_stability_1} and \subref{fig:eg_cartesian_arm_stability_2}: level sets of $H(x) + \omega (x)$ and the cutoff region $B_2(0, \mu)$.}}
    \label{fig:eg_cartesian_arm_contour}
\end{figure}

{\color{arxivcolor}We visualize the forward invariant regions given by our method and LQR in Figs.~\ref{fig:eg_cartesian_arm_forward_inv_1} and \ref{fig:eg_cartesian_arm_forward_inv_2}. As shown in Figs.~\ref{fig:eg_cartesian_arm_stability_1} and \ref{fig:eg_cartesian_arm_stability_2}, the region where $H(x) + \omega(x) > 0$ is contained in $B_2(0, \mu)$ with $\mu = 0.1$. Table~\ref{tab:eg_cartesian_arm_ratio} shows that our method achieves a larger forward invariant region than LQR. We simulate 200 trajectories starting from the boundary of $\Omega_{\widehat{V}_r}$ under the external force $d$ with $\theta_d(t) = 2 \pi t$. With the same initial conditions, the average control effort ($\int_0^T \lVert \tilde{u}(t) \rVert_2^2 dt$) of LQR and our method is 1.71 and 1.81, respectively. The average settling time ($\inf \{ \tau \geq 0 \mid \lVert x(t) \rVert_2 \leq 0.05, \forall t \geq \tau \}$) is 3.81~\si{s} and 3.44~\si{s}, respectively.}

\begin{table}[H]
\centering
\caption{{\color{arxivcolor} Forward invariant region ratio (\%) for Cartesian arm.}}
\label{tab:eg_cartesian_arm_ratio}
{\color{arxivcolor}
\begin{tabular}{cccc}
\toprule
LQR & Ours ($L_V{=}2$) & Ours ($L_V{=}3$) & Ours ($L_V{=}4$) \\
\midrule
8.62 & 20.61 & 22.58 & 24.38 \\
\bottomrule
\end{tabular}}
\end{table}

\subsection{Networked Ring System Example}
{\color{arxivcolor}
We have tested the performance of our method on the following networked ring system:
\begin{equation}
    \dot{x}_{i} = \sin(x_i - x_{i+1}) + u_i
\end{equation}
with $x_{N+1} = x_1$ and $x_i \in [-1,1]$ up to $N=8$. As shown in Table~\ref{tab:eg_ring_sys_ratio}, our method consistently achieves a larger forward invariant region than the LQR baseline across these dimensions. Although our GPU-accelerated framework substantially improves verification speeds, the inherent difficulty of high-dimensional formal verification remains an open challenge across the neural Lyapunov literature. A comprehensive formal treatment of arbitrary $N$-dimensional networked topologies is an interesting direction for future work.

\begin{table}[H]
\centering
\caption{{\color{arxivcolor}Forward invariant region ratio (\%) for the networked ring system.}}
\label{tab:eg_ring_sys_ratio}
{\color{arxivcolor}
\begin{tabular}{lcccc}
\toprule
$n_x$ & 2 & 4 & 6 & 8 \\
\midrule
LQR & 60.45 & 22.40 & 4.31 & 0.55 \\
Ours & 86.88 & 30.47 & 6.06 & 0.85 \\
\bottomrule
\end{tabular}
}
\end{table}
}

\begin{IEEEbiography}[{\includegraphics[width=1in,height=1.25in,clip,keepaspectratio]{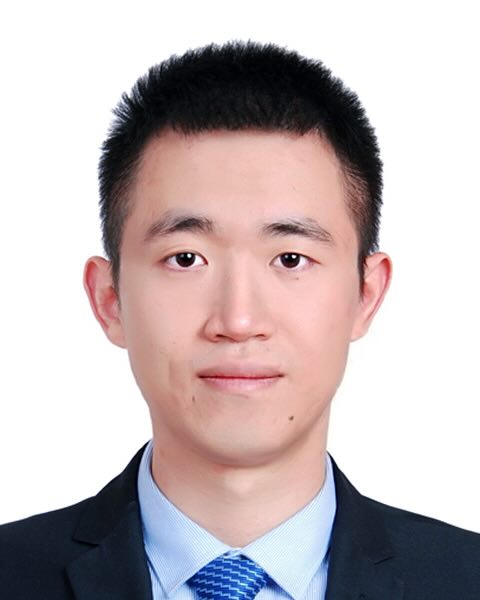}}]{Shiqing Wei} (Graduate Student Member, IEEE) received his BEng and MEng degrees in Mechanical Engineering from Shanghai Jiao Tong University, China, in 2018 and 2021, respectively. He also obtained a Dipl. Ing. with a minor in Applied Mathematics from École des Mines de Paris (now MINES Paris -- PSL), France, in 2019 through a double-degree program. Currently, he is pursuing a Ph.D. in Electrical Engineering at New York University Tandon School of Engineering in Brooklyn, NY, USA. In the summer of 2022, he joined the Control/Robotics Research Laboratory (CRRL) at NYU. His research interests include safe control synthesis, learning-based control methods, and verification of neural networks.
\end{IEEEbiography}

\begin{IEEEbiography}[{\includegraphics[width=1in,height=1.25in,clip,keepaspectratio]{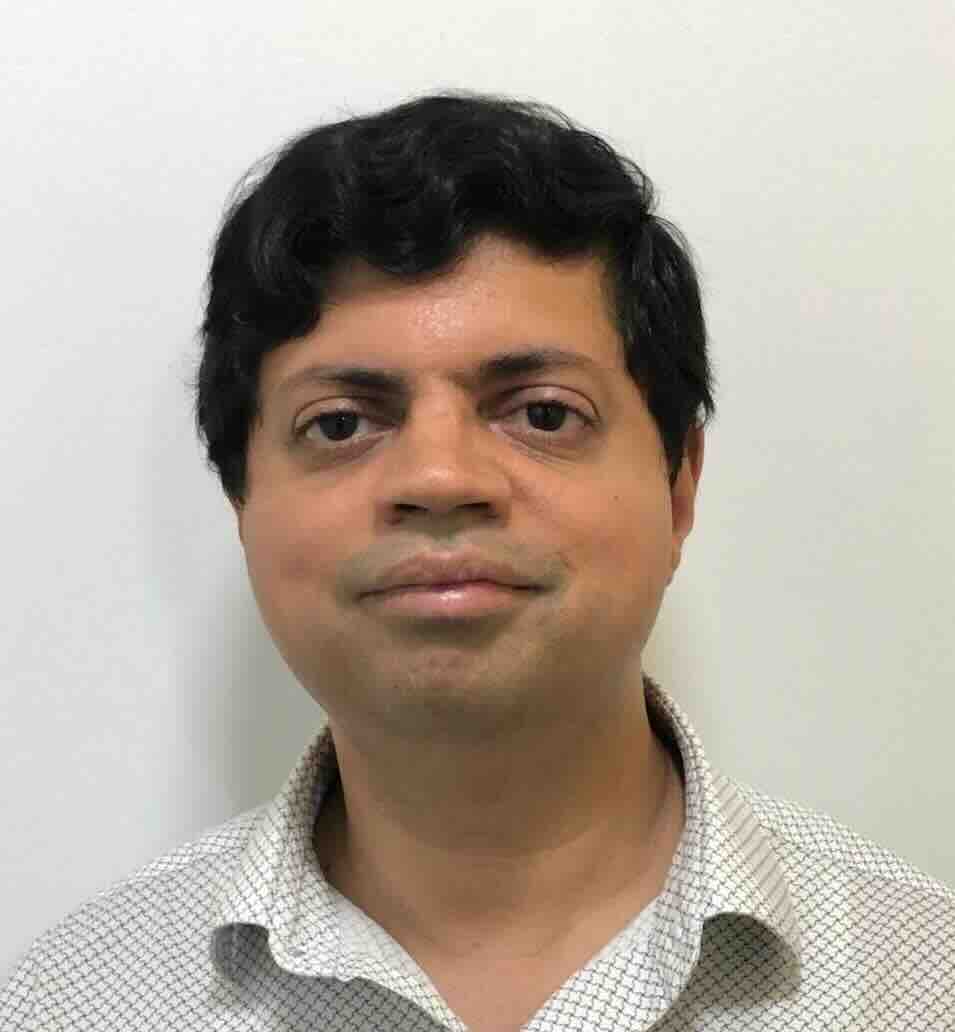}}]{Prashanth Krishnamurthy} (Member, IEEE) received his B.Tech. degree in electrical engineering from Indian Institute of Technology, Chennai, India in 1999, and M.S. and Ph.D. degrees in electrical engineering from Polytechnic University (now NYU) in 2002 and 2006, respectively. He is a Research Scientist and Adjunct Faculty with the Department of Electrical and Computer Engineering at NYU Tandon School of Engineering. He has co-authored over 175 journal and conference papers. He has also co-authored the book ``Modeling and Adaptive Nonlinear Control of Electric Motors'' published by Springer Verlag in 2003. His research interests include autonomous vehicles and robotic systems, multi-agent systems, sensor data fusion, robust adaptive nonlinear control, resilient control,  path planning and obstacle avoidance, machine learning, real-time embedded systems, cyber-physical systems and cyber-security, real-time software implementations, high-fidelity and hardware-in-the-loop simulation, and decentralized and large-scale systems.
\end{IEEEbiography}

\begin{IEEEbiography}[{\includegraphics[width=1in,height=1.25in,clip,keepaspectratio]{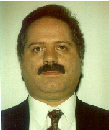}}]{Farshad Khorrami} (Fellow, IEEE) received the Bachelor's degrees in mathematics and electrical engineering in 1982 and 1984, respectively, from The Ohio State University. He also received the Master’s degree in mathematics and Ph.D. in electrical engineering in 1984 and 1988 from The Ohio State University, Columbus, Ohio, USA. 
He is currently a professor in the Electrical and Computer Engineering Department at NYU, Brooklyn, NY, where he joined as an assistant professor in Sept. 1988. His research interests include adaptive and nonlinear controls, robotics and automation, unmanned vehicles, cyber security for CPS, embedded systems security, machine learning, and large-scale systems and decentralized control. 
He has developed and directed the Control/Robotics Research Laboratory at Polytechnic University (now NYU) and is the Co-Director of the Center for Artificial Intelligence and Robotics (CAIR) at NYU Abu Dhabi.
Dr. Khorrami has also commercialized UAVs as well as developed autopilots for various unmanned vehicles. His research has been supported by the DOE, ARO, NSF, ONR, DARPA, ARL, AFRL, and several corporations. He has served as a conference organizing committee member of several international conferences. 
\end{IEEEbiography}

\end{document}